%% file: main.tex
\documentclass[11pt]{article}
\input{setup}
\begin{document}
%  -------------------------------------------------------------------------
\title{Keeping up with ``The Joneses'': reference dependent choice with social comparisons}

\author{Alastair Langtry\footnote{University of Cambridge and King's College, Cambridge CB2 1ST, UK.  Email: \emph{atl27@cam.ac.uk.} } \footnote{I am grateful to Matthew Elliott, Lukas Freund, Edoardo Gallo, Charles Parry and Gabriela Stockler for insightful comments and discussions, and to seminar participants at the University of Cambridge. This work was supported by the Economic and Social Research Council [award reference ES/P000738/1]. Any remaining errors are the sole responsibility of the author.}}

\date{28 June 2022 \\ \vspace{7mm}}

\maketitle
%  -------------------------------------------------------------------------
\begin{abstract}
\input{0.abstract}

% D60 Welfare economics: General
% D85 Network Formation and Analysis: Theory
% D91 Role and Effects of Psychological, Emotional, Social, and Cognitive Factors on Decision Making
% J24 Human Capital • Skills • Occupational Choice • Labor Productivity
% Z13	Economic Sociology • Economic Anthropology • Language • Social and Economic Stratification
\end{abstract}
%  -------------------------------------------------------------------------
\newpage
\input{0.v8_intro}

\section{Model}\label{sec:model}
\input{1.model}

\section{Social Comparisons, Consumption and Welfare}\label{sec:results}
\input{2.comp_statics}

\section{Choosing your friends: Endogenous Network Formation}\label{sec:endo network}
\input{3.endogenous_network}

\section{Choosing your coworkers: Labour Market Sorting}\label{sec:sorting model}
\input{4.labour_v2}

\section{Conclusion}\label{sec:conclusion}
\input{5.conclusion}

%  -------------------------------------------------------------------------
\clearpage
\singlespacing
\bibliographystyle{abbrvnat}
\addcontentsline{toc}{section}{References}
\bibliography{general_bib}
\cleardoublepage
\onehalfspacing
\appendix
\numberwithin{equation}{section}
\numberwithin{thm}{section}
\numberwithin{prop}{section}
\numberwithin{defn}{section}
\numberwithin{rem}{section}
\numberwithin{lem}{section}
\numberwithin{cor}{section}
%  -------------------------------------------------------------------------
% List of Appendices
% A: Proofs -- slim these down in line with Akhil's suggestions. Currently 10 pages. Aim for 5-6.
% B: A comment for policymakers
% C: Non-linear sub-utility (STILL TO DO)
% D: Loss aversion -- for the proof here. Can basically just have a one-liner saying that by inspection, agent is above ref point, so the loss aversion part does not bite. Effectively adding a non-binding constraint.
% E: Many goods
% F: Network structure
% G: Heterogeneous costs
% H: Variant on labour markets model (STILL TO DO)
% I: Formalising claims made in main text -- make proofs available upon request for submission version.
% J: Homophily and inequality -- drop from submission version.
% K: A statement of Chang's (2006) result

\newpage
\section*{Proofs}\label{proofs}
\input{A1.proofs_consumption}
\input{A2.proofs_labour}
%\input{A3.proofs_misc}

\newpage
\vspace*{\fill}
\hfill
\begin{center}
End of main text. \\
This page is intentionally left blank.
\end{center}
\vspace{\fill}

\newpage
\section*{For Online Publication}\label{extensions}

\section{A Comment for Policy-makers}
\input{B.no_network_effect}

\newpage
\section{Non-linear sub-utility}\label{extensions:nonlinear utility}
\input{C.nonlinear_subutility}

\newpage
\section{Loss Aversion}\label{extensions:loss aversion}
\input{D.loss_aversion}

% \newpage %%% not needed here
\section{Multiple Goods}\label{extensions:K goods}
\input{E.many_goods}

\newpage
\section{Network Structure}\label{extensions:network structure}
\input{F.network_structure}

\newpage
\section{Chang (2006): A perturbation theorem}\label{extension:chang2006}
\input{K.Chang_2006}

\newpage
\section{Heterogeneous Costs}\label{extensions:heterogeneous costs}
\input{G.heterogeneous_costs}

%\newpage
%\section{Formalising claims from the main text}\label{extension:claims}
%\input{I.proofs_of_claims} %% THIS IS ONLY NEEDED ON REQUEST.

%\newpage
%\section{Homophily and Inequality} \label{extensions: homophily}
%\input{J.homophily} %%% DO NOT TYPESET THIS SECTION

%\newpage
%\section{Further Applications}\label{sec:further applications}
%\input{C.further_applications}

%\input{tikz_pictures/sortin_lemma_separated}

\end{document}

%% file: setup.tex
\usepackage{natbib}
\usepackage{amsmath, amsthm, amssymb}
\usepackage{graphicx, sidecap, subcaption}
\usepackage{caption, comment, float}
\graphicspath{{graphs/}} 	% Allows loading of image files
\usepackage{hyperref}
\hypersetup{ colorlinks, citecolor=black, filecolor=black,
	linkcolor=black, urlcolor=black }
\usepackage{cleveref} % clever referencing package
\usepackage{tikz, pgf, pgfplots, xcolor}
\usetikzlibrary{arrows, calc, positioning}
\usetikzlibrary{automata, shapes, shapes.geometric}
\usetikzlibrary{decorations.markings}
\usepackage{setspace, fullpage, pdfpages, import}
\usepackage{booktabs, xcolor}
\newtheorem{prop}{Proposition} 	%\theoremstyle{corollary}
\newtheorem{cor}{Corollary}		%\theoremstyle{lemma}
\newtheorem{lem}{Lemma}			%\theoremstyle{definition}
\newtheorem*{lem*}{Lemma}
\newtheorem{defn}{Definition}		%\theoremstyle{remark}
\newtheorem{rem}{Remark}			%\theoremstyle{definition}
\newtheorem{ass}{Assumption}

%% file: 0.abstract.tex
Keeping up with ``The Joneses'' matters. This paper examines a model of reference dependent choice where reference points are determined by social comparisons. An increase in the strength of social comparisons, even by only a few agents, increases consumption and decreases welfare for everyone. Strikingly, a higher marginal cost of consumption can increase welfare. In a labour market, social comparisons with co-workers create a \emph{big fish in a small pond} effect, inducing incomplete labour market sorting. Further, it is the skilled workers with the weakest social networks who are induced to give up income to become the \emph{big fish}. \\ % Version 7. 100 words

%% file: 0.v8_intro.tex
\begin{quote}
    \emph{}
\end{quote}

\begin{quote} \emph{\footnotesize{
``No man is an island, entire of itself'' }}
\flushright{
\footnotesize{Meditation XVII, John Donne, 1624}}
\end{quote}
\vspace{3mm}

% Para 1: Opening gambit
Working life has long provided ample opportunity to compare yourself to your coworkers: from the clothes they wear and the cars they drive to their latest luxury holiday. Trends in clothes, cars and holidays have come and gone, but the ability to compare yourself to your coworkers has not. In contrast, the rapid growth of social media over the past two decades has made it much easier to compare yourself to your friends. Photos and videos of their houses, holidays and parties are far more readily available than ever before. A short scroll through Facebook or Instagram now shows even your loosest of acquaintances' consumption decisions -- people who you would previously have never been able to compare yourself to.

% What I do 
People do not make decisions in a vacuum. Instead, they often evaluate their actions relative to the world -- and the people -- around them. Taking this casual insight seriously, I build a model of reference dependent choice where agents derive benefits from consumption relative to that of their friends.
To do this, I take the benchmark model of \cite{koszegi2006model} and change the determinant of the reference point. In \cite{koszegi2006model}, the reference point is an agent's beliefs about her own future consumption. In contrast, I assume the reference point is a weighted sum of other agents' consumption, where the weights capture the strength of social comparisons. 

% The networks bit
To incorporate social comparisons it is necessary to extend \cite{koszegi2006model} to include multiple agents. Agents are embedded in a social network, and their reference point is derived from their neighbours. There is also a constant marginal cost of consumption and a constant marginal of benefit of connections in the social network -- to capture the myriad benefits of social connections. Agents make their consumption decisions simultaneously and I study the Nash equilibria of this game.

I first consider a setting where the network is fixed, and agents choose consumption directly. With a fixed network, a unique equilibrium exists when social comparisons are not too strong. Equilibrium consumption (when it exists) relies critically on Bonacich centrality -- which captures how well-connected agents are to other well-connected agents.\footnote{\cite{katz1953new} and \cite{bonacich1987power} originally proposed this metric. The link between equilibrium consumption and Bonacich centrality is due to \cite{ballester2006s}.} Within the confines of the model, stronger social comparisons, even by only a few agents, increase consumption but \emph{reduce} welfare for all agents. Higher consumption by one agent pushes her neighbours to consume more, through an increase in their reference points. This creates a rush to ``keep up with The Joneses''. Everyone spends more on consumption, but nobody gets any further ahead, so welfare falls. 

Strikingly, a higher marginal cost of consumption can increase welfare. This arises because it reduces \emph{other people’s} consumption, reducing the need to ``keep up with The Joneses'' and making makes an agent better off for a given level of her own consumption. When an agent’s Bonacich centrality is sufficiently high, this effect dominates the standard effects that cause higher costs to reduce welfare. Within my benchmark model, a flat tax would improve then welfare, even if the proceeds are entirely wasted.

% Para 6
Next, I examine the implications of agents choosing their neighbours (i.e. when the network is endogenous). In this case, they separate into a set of ``social classes''. This is where agents are only friends with people who have the same level of consumption as they do. 

%I then extend the model to build in a labour market with supermodular output. Agents compare their income to both their friends and their co-workers.\footnote{The idea that co-workers' income affects agents' utility has a long history in economics \citep{marshallprinciples, veblen1899} and is supported by much more recent evidence in a range of settings (see for example \cite{card2012inequality, breza2018morale}).} 
%This allows us to explore how social comparisons affect labour market sorting and segregation. Sorting in this context is where skilled workers work disproportionately for high productivity firms, and segregation is where skilled workers work disproportionately with each other. 

One way of interpreting the cost of consumption is the effort of work needed to acquire it. %\footnote{It doesn't matter whether this is effort of working to earn the money that buys consumption, or as the effort of producing consumption goods directly.} 
Such a view implicitly thinks of work as a solitary activity. In reality however, most people work in firms. It is typically possible to increase consumption by moving to a more productive firm, without necessarily working harder. However, moving firms comes with an obvious catch -- your coworkers change. Therefore we can also think of the change in coworkers' consumption as the cost of moving firm to get higher consumption for yourself. This is unavoidable.\footnote{The idea that coworkers' income affects agents' utility has a long history in economics \citep{marshallprinciples, veblen1899} and is supported by much more recent evidence in a range of settings (see for example \cite{card2012inequality, breza2018morale}).} %\footnote{Some readers may view this as perhaps \emph{the} key attraction of changing firms. But such a stance is material for an entirely different paper.} 

To recognise this reality I then allow agents to choose their firm. This allows us to explore how social comparisons affects labour market sorting and segregation. Sorting in this context is where skilled workers work disproportionately for high productivity firms, and segregation is where skilled workers work disproportionately with each other. I consider a stylised benchmark model with two types of workers and two types of firms.
The model is then able to explain two key stylised facts that have emerged from empirical studies of labour markets. First, that sorting is not complete -- the most skilled workers are frequently not matching with the most productive firms. Measuring sorting as the correlation between firm and  worker types, many authors find a only moderate degree of labour market sorting in practice -- a correlation coefficient in the region of 0.4 to 0.6. \cite{bagger2019empirical} find this level of sorting in Danish data, \cite{borovivckova2017high} in Austrian administrative data, and \cite{torres2018sources} and \cite{card2018firms} in Portuguese data. 

Social comparisons with coworkers can explain the incomplete sorting observed in practice. Given the stark assumptions of the model, they can cause a skilled worker to want to be the \emph{big fish in a small pond}. This means she wants to work with coworkers who earn less than she does, even if her own wages are lower. Consequently, some skilled workers choose to work at low productivity firms in order to work with lower paid coworkers. 

The second stylised the fact the model can account for is that sorting has increased over time. \cite{song2019firming} and \cite{card2013workplace} demonstrate that sorting and segregation have risen over time, and both attribute the majority of the increase in income inequality over recent decades to these phenomena. The model predicts increasing sorting when there is technology change that facilitates better comparisons with friends. This is because better comparisons with friends causes the big-fish effect to diminish over time, in turn increasing sorting. 

% Para 11
I argue that the the rise of social media is precisely such a technology change. Social media has made it far easier to observe the consumption decisions of a wide group of friends and acquaintances. There is also increasing empirical evidence that social media intensifies social comparisons \citep{burke2020social, verduyn2017social, verduyn2015passive, krasnova2013envy}, although it is difficult to clearly demonstrate causality. The huge increase in social media use over recent decades would then show up in my model as stronger social comparisons with friends. This makes the \emph{big fish} effect weaker and so drives up sorting in the labour market, as workers find it more important to keep up with their friends rather than to be the \emph{big fish} at work. More recently, the increasing shift to working from home in some industries may make social comparisons with friends more important relative the those with co-workers by causing workers to see their coworkers less frequently. This shift could also exacerbate the effects of social media -- driving sorting up further. 

Beyond the extent of sorting, the model also predicts the \emph{pattern} of sorting: that is, who works where. Among skilled workers, it is those making the weakest social comparisons with friends who choose the work for low productivity firms and earn higher wages relative to their co-workers, but lower \emph{absolute} wages. This presents a novel mechanism through which social connections increase wages. With weak social connections, the \emph{big fish} effect looms large, and skilled workers choose jobs at low productivity firms. They are not crowded out by those with better social connections, but feel pressured into low productivity firms by a strong \emph{big fish} effect. 
Hence workers with weaker social networks \emph{choose} lower productivity jobs, rather than having worse job opportunities due to fewer referrals (see, for example, \cite{bolte2020role}). Which of these mechanisms -- social comparisons or job referrals -- is at work has important policy implications. If social comparisons play the leading role, then Bolte et al.'s [2020] main policy proposal for addressing inequality -- one-time affirmative action to provide jobs to disadvantaged groups -- will be less effective. In a world where social comparisons are important, social networks create inequality of \emph{aspiration}. Weak social networks cause skilled workers to \emph{prefer} jobs at low productivity firms via the \emph{big fish} effect. 

% Para 14
Addressing inequality of opportunity -- the goal of Bolte et al.'s [2020] affirmative action -- is not sufficient to ensure that skilled workers from disadvantaged groups receive the high productivity jobs. In contrast, strengthening the social networks of skilled workers from disadvantaged groups addresses both issues. It provides them with better job opportunities (via more referrals) and also changes their preferences so that they choose to work in high productivity jobs. While altering social networks is very challenging for a policymaker, it may be necessary if currently disadvantaged groups are to achieve equal labour market outcomes.

\newpage
\subsection*{Related Literature}
%\noindent \textbf{Related literature.} 
% Para 1
This paper continues a strand of the behavioural economics literature that examines the source of reference points.\footnote{There is of course a much wider literature concerning reference dependence and loss aversion, with seminal contributions from \cite{tversky1979prospect, thaler1980toward, tversky1991loss}. See \cite{barberis2013thirty} for a review of the literature, and \cite{dellavigna2009psychology, camerer1998prospect} for reviews of some of the empirical evidence.} 
%To my knowledge, it is the first to explicitly consider reference points determined by social comparisons (or, conversely, to motivate social comparisons as acting through reference points).
\cite{koszegi2006model} provide perhaps the best known contribution to this strand. They suppose that an agent's reference point depends on her expectations and find that in situations with uncertainty, an agent is willing to pay more for a good when she believes she is more likely to buy it. This paper considers a setting without uncertainty, and so one where K{\H{o}}szegi and Rabin's mechanism has no bite.

Beyond my novel motivation for including social comparisons in agents' preferences, this paper makes two main contributions to the significant literature on social comparisons.\footnote{Within economics, this strand dates at least to \cite{veblen1899}. Other classical contributions include \cite{duesenberry1949income, layard1980human, frank1985choosing, frank1985demand, clark1996satisfaction, cres1997externalities} and \cite{clark1998comparison}. More recent contributions include \cite{hopkins2004running, luttmer2005neighbors, frank2005positional} and \cite{schelling2006micromotives}.} 
First, is the striking result that a uniform increase in the marginal cost of consumption can penalise agents who are highly central in the social network while benefiting less central agents. There is an obvious and potentially important policy implication. A flat tax on consumption makes the rich better off. 

This builds on Frank's [\citeyear{frank1985choosing, frank2001luxury}] idea that progressive consumption taxes can improve outcomes by ``slowing down'' the race to keep up with the Joneses. Making status goods more expensive reduces everyone's consumption of them. But it does so in a way that doesn't change people's \emph{position} in terms of consumption levels. So everyone maintains their position, but consumes less of the status goods -- making them better off. The core logic of my result mirrors Frank's argument. Introducing a network highlights that the impacts of government intervention can depend critically on an agent's position in a social network, and that the Frank's idea about progressive consumption taxes extend to flat taxes.\footnote{This idea of extending to a flat tax is not entirely novel. \cite{frank1985choosing} argues that government regulation can prevent a race to the bottom on workplace safety standards. This is akin to a flat tax in this context. Pressure to ``keep up with The Joneses'' can induce workers to accept worse conditions in exchange for some additional income. But because everyone faces this pressure, nobody ``gets ahead of The Joneses'' and are all worse off as a result.} It also allows for richer, or at least more explicit, patterns of social comparisons.

%This also complements Frank's [\citeyear{frank1985choosing}] idea that government regulation can prevent a race to the bottom on workplace safety standards. \cite{frank1985choosing} notes that pressure to ``keep up with The Joneses'' can induce workers to accept worse conditions in exchange for some additional income. But because everyone faces this pressure, nobody ``gets ahead of The Joneses'' and are all worse off as a result.
%However, my result is somewhat more subtle as it shows that the impacts of government intervention can depend critically on an agent's position in a social network. The explicit social network also allows for richer patterns of social comparisons. It is also broader in the sense that it applies directly to the cost of consumption.

Further, it complements the extensive literature on the Easterlin Paradox \citep{easterlin1974does}.\footnote{See \cite{clark2008relative} for a thorough survey. Some more recent work in this area includes \cite{decancq2015happiness} and \cite{easterlin2020easterlin}. Easterlin's original finding is that individual happiness correlates positively with income at a point in time, but aggregate happiness does not rise over time as aggregate incomes increase.} With the confines of my model, I find that higher income (captured by my model as a lower marginal cost of consumption) can reduce welfare. This is even starker that Easterlin's finding of little-to-no relationship between income and happiness over time. This is because the only good I have in my model is a status good.

Second, it adds to the branch that focuses on labour market behaviour.\footnote{Some important empirical contributions include \cite{perez2011if, card2012inequality, breza2018morale} and \cite{fu2019social}. They demonstrate that social comparisons with coworkers and are important for labour market behaviour. My contribution is theoretical.} Existing literature shows that social comparisons impact a variety of outcomes in labour markets -- most notably on wage-setting, both in competitive \citep{frank1985demand}  and imperfectly competitive environments \citep{goerke2017social}. I explore the impact of social comparisons on labour market sorting. 

% Para 6
Close to my paper, \cite{von2012social} shows how they can induce workers to sort into unemployment. \cite{von2012social} assumes there are two types of workers, and only one type make social comparisons. He additionally makes a very stark assumption that unemployed workers only compare themselves to other unemployed workers, while the employed compare themselves to both others at their firm and workers at other firms. This asymmetry in the composition of social comparisons then drives some to choose unemployment. Related, \cite{vasquez2021} shows how co-worker altruism might induce wage rigidity and hence involuntary unemployment.

In contrast, all agents in my model make social comparisons and their non-coworkers do not change when they move to a different firm. Additionally, social comparisons in my model can explain why labour market sorting is incomplete and why it has risen over recent decades.

My model is a network game with linear best replies, by now a well-studied class of games.%It is this property, coupled with strategic complementarities, that yields the link between equilibrium outcomes and Bonacich centrality \citep{ballester2006s, bramoulle2014strategic}. 
\footnote{A review of this strand of literature is beyond the scope of this paper. See \citet*[Chapter 5]{bramoulle2016oxford} and \citet*[Chapter 9]{jackson2008} for comprehensive coverage.} I make two contributions to this literature. First, I to show a comparative static for rerouteing a link (extending well-known results concerning adding or removing a link). Second, I demonstrate one way in which network games can be used to model a well-known behavioural bias. Some of the first results, notably \Cref{prop:equilibrium} and \Cref{prop:ref strength}, are similar in technical detail to results in \cite{ballester2006s} and \cite{ushchev2020social} respectively, but their motivation via reference dependent choice is novel.

There is also a growing strand of the networks literature that examines endogenous network formation. Perhaps closest to my paper, \cite{hiller2017peer} considers a network game with strategic complementarities where links are two-sided (and so require consent from both agents), finding that equilibrium networks are either complete or empty. A key difference in this paper is that social comparisons are a negative externality (\cite{hiller2017peer} assumes positive externalities).\footnote{Related, \cite{staab2019formation} studies how agents who make social comparisons form groups when there is a price to group membership. In this setting, he finds that social comparisons can reduce -- rather than increase -- segregation.} %\footnote{Relatedly, \cite{galeotti2010law} study information flow using a model where links can be formed unilaterally, and only one agent pays the cost of forming the link. \cite{galeotti2010law} find that this leads to a form of core-periphery networks where each agent in the periphery forms a link with every agent in the core.} 
In contrast, I find that agents form groups based on their benefits of forming links, and only link with others in the same group. This creates groups, each with a different level of consumption, who do not interact with one another -- very stark social classes emerge. %\footnote{In the case where all agents have the same benefits of forming links, there is only one group, and the result specialises to something close to \cite{hiller2017peer} -- although due to the weighted network, agents need not form the complete network.} 

Related, \citet{ushchev2020social} study a game where agents want to be close to the average action in their neighbourhood. They suppose some preference for conformity. In spite of the contrasting motivation and the simplified network component, \citet[S5.5]{ushchev2020social} also find that the network becomes completely homophilous when agents can choose their links. This means agents only form links with others of the same type (their analysis assumes only two types for simplicity).\footnote{More recently, \cite{sadler2021games} have nested a wide variety of endogenous network formation games into a single framework. While my game does not fit because I consider weighted networks, the stark social classes emerging in this paper are very similar to the \emph{ordered overlapping cliques} networks in \cite{sadler2021games}.} %\footnote{This is because, to use Sadler and Golub's \citeyear{sadler2021games} taxonomy my game has ``negative spillovers'' and ``actions and link are complements''. The former means an agent taking a higher action is less attractive to potential neighbours. In my game, it is less attractive to form a link with someone who has higher consumption because the associated social comparisons are less favourable. The second property is present because the marginal benefit of consumption is higher when an agent makes stronger social comparisons (i.e. has more links).} 

% Other contributions to the endogenous networks literature include: Baetz (2015, TE), Belhaj et al (2014, GEB) [kinda], Allouch (2015, JET) [maybe]

% Para 11
In spirit, this paper is most similar to \cite{ghiglino2010keeping} and \cite{immorlica2017social}, which both examine social comparisons taking place on a network.\footnote{Few treatments of social comparisons prior to \cite{ghiglino2010keeping} include an explicit network, instead focusing on average consumption (e.g. \cite{aronsson2015keeping, hopkins2004running}) or an agent's ordinal rank (e.g. \cite{neumark1998relative, frank1985choosing, frank1985demand}).} %To my knowledge, none of the literature so far has motivated the effects of social comparisons as acting through an established behavioural bias. Rather, social comparisons are assumed to enter agents' utility functions directly.}  
\cite{ghiglino2010keeping} consider a two-good general equilibrium exchange economy, where social comparisons apply to one good. In both this paper and in \cite{ghiglino2010keeping} Bonacich centralities are the key determinant of agents' consumption.\footnote{\cite{ghiglino2010keeping} also find that Bonacich centralities determine prices -- something outside the scope of this paper.} 
My focus on a single good and a partial equilibrium approach permits my stronger comparative statics results regarding consumption and welfare. It also facilitates investigation of endogenous network formation and the application to labour markets. Further, it also allows me to drop two assumptions \cite{ghiglino2010keeping} make for tractability: that an agent cares about her consumption relative to the simple average of her friends' consumption, and that the strength of social comparisons only depends on the number friends. Ghiglino and Goyal's \citeyear{ghiglino2010keeping} use of Cobb-Douglas utility helps generate linear best replies in a setting with multiple goods. I make alternative, and even starker, simplifications -- namely a single good and linear costs  %Further, I can drop three assumptions \cite{ghiglino2010keeping} make for tractability: a Cobb-Douglas utility function, that an agent cares about her consumption relative to the simple average of her friends' consumption, and that the strength of social comparisons only depends on the number friends. 

\cite{bramoulle2022loss} subsequently build loss aversion into the framework of \cite{ghiglino2010keeping}, and find that loss aversion has a material impact on outcomes. It can generate a continuum of equilibria where all agents consume the same quantity of the status good. This occurs when agents' incomes are sufficiently similar to each other (relative to the strength of loss aversion). In contrast, agents in my model always consume above their reference point (even if only by a small amount), and so loss aversion has little bite. This is because they do not face an immovable budget constraint, but rather a constant marginal cost of consumption.

In its setup, \cite{immorlica2017social} is closer to this paper -- it also takes a partial equilibrium approach and focuses on a single good. The key difference is that \cite{immorlica2017social} suppose that agents only make social comparisons with to those richer than themselves. This comes at a significant cost to tractability. They find multiple equilibria, and only analyse one -- the one with highest consumption by agents. There, they find that agents stratify into a ``class structure'' according to their position in the network. Consumption then depends on ``social class''. In contrast, I assume agents form an overall reference point based on their friends' consumption. This preserves continuity in most functions, provides a unique equilibrium and allows for simpler determination of equilibrium consumption and welfare, and associated comparative statics. I also find an analogy to Immorlica et al.'s (2017) ``social classes'' when allowing agents to choose their own friends.

%In spite of its different motivation, my model is very similar in its set-up to the study of local public goods by \cite{bramoulle2007public}. The key changes are that I assume (1) strategic complementarities rather than strategic substitutabilities, (2) a weighted and directed network, and (3) exogenous benefits to connections. Intuitively, we can think of this change in the sign of the externality as a switch from local public goods to ``local public bads''. %This aligns with my finding that stronger social comparisons reduces welfare. 

% Para 13
The rest of the paper is organised as follows. \Cref{sec:model} sets out the model. \Cref{sec:results} presents the main results when the network is fixed. \Cref{sec:endo network} considers how agents choose friends. \Cref{sec:sorting model} considers how agents choose coworkers and jobs. \Cref{sec:conclusion} concludes. 

%% file: 1.model.tex
%This section sets out the agents, their actions, their reference points and their preferences. It then discusses Bonacich centrality -- a measure of network centrality that will be important throughout the paper. \\

There are $n$ agents. Each agent $i$ simultaneously chooses a level of consumption $x_i \geq 0$. The agents are embedded in a weighted and directed network $G$. The network $G$ is an $n \times n$ matrix, where the entry $G_{ij}$ denotes the strength of a link from $i$ to $j$. An agent $j$ is $i$'s \emph{neighbour} (or \emph{friend}) if and only if $G_{ij} > 0$. Assume $i$ cannot be neighbours with herself (so $G_{ii} = 0$ for all $i$).

Now define $\alpha_{i} = \sum_{j} G_{ij}$ and $g_{ij} = \frac{G_{ij}}{\alpha_{i}}$ for all $i,j$. Then $G \equiv \text{diag}(\alpha) g$, where $\alpha = [\alpha_1,...,\alpha_n]^{T}$. Notice that $g$ is row stochastic (i.e. each row sums to one). Call $\alpha$ the \emph{strengths} and $g$ the \emph{structure}. This novel decomposition will help us to better understand the network effects later on.\footnote{Mathematically, this is a trivial decomposition. Conceptually however, it allows us to separate out the effects of social comparisons per se from the network effects.}

An agent $i$'s \emph{reference point} is a weighted sum of her neighbours' consumptions, where the weights are determined by the network $G$. Mathematically, $i$'s reference point is $\alpha_i \sum_{j} g_{ij} x_{j}$. 
An agent $i$ cares about her consumption, $x_i$, only relative to her reference point, not in absolute terms. All agents face a constant marginal cost of consumption, $c> 0$, and receive a marginal benefit $b_i > 0$ of links. So $i$'s utility function is:

\begin{align}\label{eq: prefs}
u_{i} = f \left( x_{i} -  \alpha_{i} \sum_{j} g_{ij} x_{j} \right) - c x_{i} + b_i \alpha_i \sum_j g_{ij}
\end{align}

where $f(\cdot)$ is twice continuously differentiable, strictly increasing and concave. To focus on the interesting where agents would choose positive but finite consumption absent any social comparisons, I assume that $f'(0) > c > f'(+\infty)$. Further, assume that for any $i$, there exists some $j$ such that $b_i = b_j$. That is, every agents shares her value of $b_i$ with at least one other agent. \\

\noindent \textbf{Discussion.} The key feature of reference dependent choice here is that an agent get benefits from consumption \emph{relative to} her reference point, but pays costs for absolute consumption. The key innovation relative to existing work on reference dependent choice is a novel source of the reference point -- comparisons with neighbours in a social network.
Several of the assumptions I make are for clarity of exposition only. I relax each in the Online Appendix. Doing so does not change the core insights of the model.
\begin{itemize}
    \item \textbf{Linear sub-utility.} In \cite{koszegi2006model}, the benefit of consumption relative to the reference point (which they call ``gain-loss utility'') is a function of agents' \emph{utility from} consumption. Formally, it is $f(m(x_i) - m(r_i))$, where $r_i$ is $i$'s reference point and $m(\cdot)$ is a concave function. I assume this sub-utility function, $m(\cdot)$ is linear. %\footnote{This linearity assumption is the approach taken by \cite{tversky1991loss} in their seminal model.} 
    
    Relaxing this, and using some strictly concave function $m(\cdot)$ in \Cref{eq: prefs} does not change the insights of the model, but makes them harder to follow. However, it does prevent closed form solutions, and so breaks the tight link between network centrality and outcomes. %See \Cref{extensions:nonlinear utility}.

    \item \textbf{Single good.} Moving to a K-dimensional consumption bundle with additively separable preferences, as in \cite{koszegi2006model}, would have no impact on the results. Here, this means that agents make optimal choices for each good separately. So a K-dimensional bundle collapses down to K separate single-good problems. %\Cref{extensions:K goods} shows that the model can be generalised to a K-dimensional bundle without difficulty, but that doing so yields no additional insight.\footnote{\cite{koszegi2006model} assume that each good in the consumption bundle is additively separable. Here, this means that agents make optimal choices for each good separately. So a K-dimensional bundle collapses down to K separate single-good problems. Therefore it has no impact on the results.}
    
    \item \textbf{Homogeneous cost.} My model, embodied by \Cref{eq: prefs}, assumes that the marginal cost of consumption is the same for all agents. Introducing arbitrary heterogeneity in the cost parameter, $c_i$, does not materially change the results, but complicates matters without providing additional insight. %See \Cref{extensions:heterogeneous costs}.
    
    \item \textbf{No loss aversion.} In \Cref{eq: prefs} , the function $f(\cdot)$ is everywhere concave. Capturing loss aversion \`a la \cite{tversky1979prospect} involves making $f(\cdot)$ kinked at zero and convex in the negative domain. In my model adding this has no impact, as agents always choose to consume above their reference point (even if only by a very small amount).
    %\footnote{In the Online Appendix, I extend the functional form of $f(\cdot)$ to capture \emph{loss aversion} as well as reference dependence, in line with the Kahneman and Tversky's canonical setup \citep{tversky1979prospect, tversky1991loss}. This extension has little impact.} %In practice, this means $f(\cdot)$ is concave [resp. convex] in the positive [resp. negative], and kinked at zero. However, in our model, agents will always choose to consumer above their reference point, because if an agent is consuming below her reference point, then the marginal benefit of additional consumption exceeds the marginal cost.} 
\end{itemize}

\noindent Additionally, the linking benefit, $b_i$, captures the simple idea that social connections bring benefits. These include improved health, job opportunities, access to new information and risk sharing, to name but a few.\footnote{\cite*{reynolds1990social} show that social connections increase cancer survival rates, and \cite{newman2020value, gerst2015loneliness, kawachi2001social} document mental health benefits. \cite{calvo2004effects, chandrasekhar2020} show the roles social networks can play in getting a job, and \cite{banerjee2013diffusion} shows improved access to other information. \cite{karlan2009trust, ambrus2014consumption} demonstrate how networks can facilitate risk sharing.} %Simple introspection should also convince all but the most miserly of readers of the benefits of social connections.} 
In the interest of simplicity, I assume that there are no benefits from friends-of-friends. Only direct friends bring benefits. Implicitly this assumes indirect connections are not of first-order importance. In any case, these benefits only play a role when the network is endogenous. \\

\noindent \textbf{Bonacich centrality.} An agent cares about her neighbours' consumption, who in turn care about their neighbours' consumption, and so on. She will end up making comparisons, albeit indirectly, with everyone she is \emph{connected} to in the network. Agent $i$ is connected to $j$ if there exists a \emph{walk} from $i$ to $j$.  A walk is a sequence of links $g_{i,i+1},...,g_{j-1,j}$ such that $g_{k,k+1} > 0$ for all links in the sequence. 

The strength with which $i$ compares herself to $j$ depends on the weight of these walks (of any length) from $i$ to $j$. Bonacich centrality \citep{bonacich1987power}, counts the total weight of all walks (of any length, from 1 to $\infty$) from agent $i$ to \emph{all other agents}. It captures the idea that making strong comparisons with society as a whole involves making strong comparisons with those who themselves make strong comparisons.

\begin{defn}[Bonacich Centrality] \label{defn:bonacich centrality}
\noindent The Bonacich centrality of agent $i$ is:
\begin{align}
C^{b}_{i} = \sum_{j} \left[ \sum_{k=0}^{\infty} G^k \right]_{ij}
\end{align}
%= \sum_{j} (I - G)^{-1}_{ij}
\end{defn}

Bonacich centrality is typically defined as $C^{b}(G,\beta) = \sum_{k=0}^{\infty} \beta^k G^k$, where $\beta$ is a decay parameter (see, for example, \citet[Ch9]{jackson2008}). In my model however, the choice of decay parameter is determined by individual preferences as it is implicitly built into the reference strength. %Defined this way, Bonacich centrality is in fact a family of measures, one for each value of $\beta$. The values, and even ordinal rankings, of the centralities depend on the choice of $\beta$. A drawback of the typical definition is that there is often little to motivate the choice of any particular value of $\beta$, except that it must be chosen to ensure the measure is well-defined. In my model, however, the choice of decay parameter is determined by individual preferences, as it is implicitly built into the reference strength. So I have a single measure -- not a whole family -- and find equilibrium existence purely as a result of preferences.

%% file: 2.comp_statics.tex
% Para 1
The natural starting point is to ask how agents behave in equilibrium. Before turning to the result, it is helpful to see an agent's best response function.
\begin{align}
    BR(x_{-i}) = f^{' -1}(c) + \alpha_i \sum_j g_{ij} x_j
\end{align}
Note that best responses are linear in neighbours' consumption. I now show that there is a unique equilibrium set of actions, where consumption is proportional to agents' Bonacich centrality.

\begin{rem}\label{prop:equilibrium}
If $\lambda_1 < 1$ then there is a unique Nash Equilibrium with $x_{i}^{*} = C^{b}_{i} \cdot f'^{-1}(c)$ for all $i$, where $\lambda_1$ is the largest eigenvalue modulus of the matrix $G$.
\end{rem}

Intuitively, the largest eigenvalue $\lambda_1$ captures the extent to which a change is amplified by the network. When $\lambda_{1} < 1$, the impact of an agent's consumption on that of her neighbours is (eventually) less than one-for-one. A sufficient condition for this is $\alpha_{i} < 1$ for all $i$.\footnote{$\alpha_i < 1$ for all $i$ implies that $G$ is a substochastic matrix. So $\lambda_1 < 1$ follows from the Perron-Frobenious Theorem (see for example \cite[Ch.1]{seneta2006non} or \cite{pillai2005perron}).} In that case, the impact is less than one-for-one effect happens at every step. In contrast, Bonacich centrality is not well defined when $\lambda_1 > 1$, and the solution ceases to have a meaningful economic interpretation. Therefore I will assume that $\lambda_1 < 1$ throughout the remainder of the paper.

% Chat about solution
Agents with high centrality are more closely connected to other agents with high consumption. This means they have a high reference point, which in turn pushes them to choose high consumption. So consumption exhibits strategic complementarities: higher consumption by one agent raises the marginal benefit of consumption for her neighbours (via the effect on the reference point). This leads to similar outcomes to other network games of strategic complements \citep{ballester2006s, calvo2009peer}. The differences in functional form and in motivation do not matter for this. 

%\subsection{Comparative statics} %\noindent \textbf{Comparative statics.} 
% Reference Strength
I now turn to the impact of individual components of the model on consumption and welfare. The obvious first component is the strength of social comparisons.

\begin{prop}\label{prop:ref strength}
If $\lambda_{1} < 1$: (i) $x_{i}^{*}$ is weakly increasing, and (ii) $u_i^*$ is weakly decreasing, in $\alpha_{j}$ for all $i,j$, and strictly so if $i=j$.
\end{prop}

% Para 6: Discussion of consumption / welfare
When an agent feels social comparisons more keenly (i.e. her reference strength increases), she chooses higher consumption to better keep up with ``The Joneses''. This more intense social competition has a knock-on effect: an increase in $i$'s consumption raises her neighbours' reference points, pushing them to consume more. This effect ripples out through society, increasing \emph{everyone's} consumption. However this social competition is destructive. Given the assumptions of the model, agents can only ever keep up with ``The Joneses'' -- collectively they get stuck in a rat race and nobody gets ahead. So everyone becomes worse off as their benefits of consumption don't change, but they have to spend more.

% Policy implications
Coupled with the growing evidence that social media intensifies social comparisons (for example \cite{burke2020social, verduyn2017social, verduyn2015passive} and \cite{krasnova2013envy}), this suggests that social media has been harmful for agents. Such an insight is clearly not new, but this result clarifies the mechanism through which it actually harms its users. The increased importance of online social networks also provides policy-makers with more options, as they are much easier to influence than in-person networks. Online, policy interventions across the entire network are possible. For example, Instagram has trialled removing the number of ``likes'' attached to a post \citep{instachange2021, instachange2019} in an apparent attempt to limit some of these harmful social comparisons.
% COMMENT FROM CHARLES: doesn't this only matter for linking benefits, and so doesn't matter here?
% RESPONSE: The idea is that ``likes'' are social comparisons, so shutting them down reduces the strength of social comparisons. However, I agree this isn't a perfect example.

% Cost
A comparative static for the marginal cost of consumption are of clear interest. The first part of the result is obvious and of little interest: when consumption is more expensive, agents do less of it. In contrast, the second part is far less intuitive and has important policy implications.

\begin{prop}\label{prop:cost}
If $\lambda_1 < 1$: (i) $x_i^*$ is strictly decreasing and convex in $c$ for all $i$,
(ii) $u_i^*$ is strictly increasing in $c$ if and only if $\frac{C_i^b - 1}{C_i^b} > \frac{F(c)}{- F'(c) \cdot c}$ for all $i$, where $F(c) \equiv f^{'-1}(c)$.
\end{prop}

% Discussion of consumption / welfare
When an agent's Bonacich centrality is above a threshold, then she becomes \emph{better off} when consumption becomes more costly.\footnote{Note that $C_i^b \geq 1$, so the left-hand side of part (ii) is positive and increasing in $C_i^b$. $F(c) > 0$ and $F'(c) < 0$, so the right-hand side is positive.} The threshold is governed by the concavity of the benefits function, i.e. how quickly marginal utility diminishes when she increases her consumption. Increased cost has three effects. First, it increases the cost of a given amount of consumption, which lowers welfare. Second, it reduces an agent's consumption, which reduces overall cost. Third, it reduces her neighbours' consumption, allowing her to reduce her own consumption further, while still keeping up with ``The Joneses''. The strength of this final effect depends on how sensitive her consumption is to that of her neighbours. In other words, it depends on her Bonacich centrality. When Bonacich centrality is higher, a given change in marginal cost induces a larger fall in consumption, leading the effect of lower overall spending to dominate the higher marginal cost effect.

The condition in part (ii) of \Cref{prop:cost} is rather terse, and it is easier to see the intuition for a special case. Suppose that $f(a) = a^{\gamma}$. Then $u_i^*$ is strictly increasing in $c$ if and only if $C_i^b > \frac{1}{\gamma}$. Here, this link between centrality and the concavity of the benefits function is much more transparent.
% COMMENTS FROM CHARLES: give me an example of when this simplified condition is true.

% Para 11: Policy discussion
Within the confines of the model, this result identifies a condition where taxation will increase welfare, even if the government completely wastes the funds raised. This model focuses on a single good. In reality, there are many goods, and the strength of social comparisons may differ across goods.\footnote{I extend the analysis to multiple goods in the Online Appendix. But note that this generalisation relies on additive separability of different goods. Further, \cite{bramoulle2022loss} suggests that the generalisation to multiple goods is less straightforward when there is an immovable budget constraint.} Goods for which this condition is met are the ideal candidates for raising tax revenue. Note that the threshold depends on an individual agent's centrality. Therefore, it is possible that some agents (those with high centrality and low utility) benefit from an increase in $c$, while others (those with low centrality and high utility) lose out. So flat taxes may actually redistribute utility from the well-off to the badly-off (although the well-off agents in this model are those with low consumption). If a policy-maker cares about utility, then flat taxes can be progressive, even if they do not benefit everyone.

% Perturbation
The final moving part of the model is the reference structure, $g$. This is very difficult to examine in general. To help with tractability, I focus on a very simple change to the network structure (while holding the reference strength constant). Call this change a ``comparison shift''.

\begin{defn}[Comparison shift]\label{defn: elementary perturbation}
	A comparison shift is an $n \times n$ matrix $D$, where $D_{ru} = \phi$, $D_{rd} = - \phi$ for $r,u,d \in N$, and all other elements are equal to zero.
\end{defn}

Mathematically, it is an increase in one element of the matrix $G$, and a corresponding decrease in another element in the same row. Economically, this is where an agent $r$ makes \emph{stronger} social comparisons with $u$ and correspondingly \emph{weaker} with $d$.\footnote{An increase/decrease in a single element is equivalent to a change in $\alpha_{i}$, and so is covered by Proposition \ref{prop:ref strength}.} To make sure that solutions always exist, I will work in the setting where $\alpha_{i} < 1$ for all agents, and only consider perturbations that are \textit{feasible}.\footnote{The $\lambda_1 < 1$ assumption is not sufficient here. This is because a comparison shift changes the network, and so can change the largest eigenvalue.} To be feasible, we need $\phi \in [0, G_{rd}]$. That is, $r$ cannot reduce the amount she compares herself to $d$ by more than the amount she initially compares herself to $d$. 

With this definition in place, I can now show the impact of this simple change to the network.

\begin{prop}\label{prop:perturbation}
Consider a comparison shift, $D$, of magnitude $\phi$. Then: (i) $x_i^*$ is weakly increasing, and (ii) $u_i^*$ is weakly decreasing, in $\phi$ for all $i$, and strictly so if $i=r$, if and only if $C_u^b > C_d^b$. 
\end{prop}

% Discussion of consumption / welfare
The direction of the change in consumption is determined by the relative Bonacich centralities of agents $d$ and $u$. If $r$ switches to listening to a higher centrality agent (i.e. $C^{b}_{u} > C^{b}_{d}$) then everyone's consumption rises in equilibrium. This is because $r$ has a higher reference point after the change, and so adjusts her consumption upwards. This ripples out through the network as agents adjust to their neighbours' new, higher, consumption. This increase in consumption then reduces welfare (by the same logic as in \Cref{prop:ref strength}).

% Para 16: Chat about novelty
Characterising the impact of a comparison shift on consumption and welfare relies on a generalisation of the Woodbury matrix identity developed by \cite{chang2006inversion}. To the best of my knowledge, the use of matrix perturbation results of this type is novel in the economics literature.\footnote{The proof to \Cref{prop:perturbation} uses a simplified version of \cite{chang2006inversion}. For convenience, I provide a proof of simplified theorem in the Online Appendix.} %\footnote{\cite{golub1996matrix} provide further information on the Woodbury matrix identity, but it is not important for understanding the results presented here. We provide additional results regarding the impact of perturbations in the Online Appendix.} 

% Chat about limitations
A limitation of \Cref{prop:perturbation} is that it only considers a very simple change to the network structure. %However, analytic results do not extend to more complex changes. This is because the impact of any one elementary perturbation depends on the whole network immediately before it occurs. 
This is because the impact of any one comparison shift depends on the whole network immediately before it occurs. So the impact of a set of comparison shifts on a network $G$ is not simply the sum of their individual impacts. Of course, if the total ``volume'' of the shifts is small, then this naive summing up will be a close approximation, as the interaction effect between the various comparison shifts will be small.\footnote{However, it is easy to analyse the impact of arbitrary changes to the set of Bonacich centralities. For a given value of $c$, the distributions of consumption and welfare are simply linear transformations of the distribution of Bonacich centralities. %\footnote{This is simply because $x_i^* = C_i^b F(c)$ and $u_i^* = f(F(c)) - C_i^b F(c) c$.} 
The challenge is understanding how centrality responds to changes in the network, not how outcomes respond to changes in centrality. The impact of changes in centrality on aggregate welfare and inequality follow straightforwardly from this insight.}
I formalise these insights in the Online Appendix. \\

%\subsection{Discussion} 
\noindent \textbf{Discussion.}
% Wrap up discussion
The first key message for policy-makers is that consumption and welfare can move in opposite directions in this setting. Governments typically use measures of output and consumption as a proxy for welfare. GDP growth is widely sought after, and a failure to achieve it is generally considered a serious problem for an incumbent government. However, this paper suggests that seeking consumption growth -- even through innovation, improved productivity, and other measures to reduce the marginal cost of consuming goods -- may not succeed in delivering improvements in welfare. In a world with pervasive social comparisons, increasing consumption can be harmful. This means that well-meaning governments could be misguided in their efforts to improve their citizens' welfare by seeking to increase output and consumption. This key message is not novel. It dates back at least to \cite{abramovitz1959welfare} and \cite{easterlin1974does} within economics, and is closely related to Easterlin's Paradox.

% Para 19
Second, flat taxes can aid welfare by reducing how much ``The Joneses'' consume, reducing an agent's need to keep up with them. In some cases, this effect can dominate the standard downsides of higher unit costs, and overall welfare can increase. This does not require the government to do anything useful with the revenue raised -- taxation for its own sake may be a social good. This builds on Frank's [\citeyear{frank1985choosing, frank2001luxury, frank2008should}] idea that progressive consumption taxes improve welfare by ``slowing down'' the race to keep up with The Joneses. It also complements the standard approach to optimal taxation, which points to concentrating taxes in markets with standard negative externalities, or where the distortion due to taxes is low. This paper suggests that taxes should also be targeted towards goods where social comparisons are felt more keenly. In effect, this is because social comparisons act as a (non-standard) negative externality.\footnote{In the Online Appendix, I also show that if there is little correlation between the strength of agents' social comparisons and the structure of the network, then we can safely ignore the network. This is useful for policy-makers if measuring the network is too expensive or time-consuming.}

%% file: 3.endogenous_network.tex
% Para 1
I now allow agents to \emph{choose} their friends (i.e. form the network endogenously) and examine the network that emerges. This is a natural case to study because in many contexts people are able to choose who they spend time with and pay attention to, and by extension who they compare themselves to.

Assume that links here are \emph{symmetric} (that is $G_{ij} = G_{ji}$), agents can propose any links they like, but links require mutual consent to form, and can be broken unilaterally. So agent $i$ cannot become friends with $j$ unless $j$ is willing to reciprocate. Agents reach an equilibrium when nobody wants to unilaterally cut a link, and no \emph{pair} of agents want to form a link together. This is \emph{pairwise stability} \citep{jackson1996strategic}. Formally:

\begin{defn}[Pairwise stability]
A network $G$ is pairwise stable if: \\
(i) \ for all $G_{ij} > 0$: $u_i(G) \geq u_i(G - G_{ij})$ and $u_j(G) \geq u_j(G - G_{ij})$, \\
(ii) for all $G_{ij} = 0$: if $u_i(G + G_{ij}) > u_i(G)$ then $u_j(G + G_{ij}) < u_j(G)$
\end{defn}

With this definition in place, I can now show that allowing agents to form the network endogenously leads to extreme segregation.

\begin{prop}\label{prop: endogenous network}
In all pairwise stable networks, if $b_i \geq c  f^{'-1}(c)$, then $G_{ij} > 0$ only if $b_i = b_j$, and if $b_i < c  f^{'-1}(c)$ then $G_{ij} = 0$ for all $j$.
\end{prop}

The network formation process yields extreme homophily (homophily is the tendency for agents to be disproportionately connected to those similar to themselves). While highly stylised, this is consistent with the high degree of homophily typically found in real social networks \citep{mcpherson2001birds, jackson2008}. Agents segregate into a well-defined series of ``social classes'', where all the agents in a given group have the same value of $b_i/c$.\footnote{As I have suppressed heterogeneity in $c_i$ for clarity of exposition, it simply amounts to sharing a value of $b_i$ here. The ratio proper has bits when $c_i$ is heterogeneous.} Any agents whose benefit of linking is not sufficiently high (i.e. $b_i < c F(c)$) will not form any links at all -- as the costs of any social comparisons outweigh the benefits of friendships. The condition in \Cref{prop: endogenous network} that defines ``sufficiently high'' linking benefits is terse. To see it more clearly suppose that $f(a) = a^{\gamma}$. Then an agent $i$ wants to form links if and only if $b_i \geq \gamma^{1/(1- \gamma)} c^{\gamma / (\gamma - 1)}$. Even more transparently, when $\gamma = 0.5$, this amounts to $b_i \geq \frac{1}{4c}$. %Linking benefits are ``sufficiently high'' -- and an agent forms links -- if the marginal benefit of consumption exceeds the marginal cost at the point where an agent consumes $b_i/c$ above her reference point. 

An agent would ideally like to form links with others who have very low consumption, because this minimises the impact on her reference point. Further, she would like to keep forming links until the marginal benefit of doing so equals the marginal cost. However, the second effect precludes the first. An agent $i$ keeps forming links until her consumption reaches her benefit-cost ratio. However, this prevents her from forming links with some agent $j$ who has a lower benefit-cost ratio, because $j$ would not want to form a link with $i$ (as $i$ will have higher consumption than $j$). %Overall, nobody wants to form a link with someone who has higher consumption than themselves. Because links require mutual consent, this means that agents can only form links with others who have the same level of consumption.

% Para 6
Each ``social class'' has a different level of consumption. Those with higher marginal benefits of links, relative to the marginal cost of consumption, consume more. 

\begin{cor}\label{cor:endogenous network}
In all pairwise stable networks, $x_i^* = \max \{ \frac{b_i}{c} ,  f^{'-1}(c) \}$, and $C_i^b = \max \{ \frac{b_i}{c \cdot  f^{'-1}(c)} , 1 \}$.
\end{cor}

This makes it meaningful to think of ``higher'' and ``lower'' social classes. While each group has no distinguishing feature in terms of network structure, they are ordered by the amount that their members consume. In the ``higher'' social classes, agents make stronger social comparisons within their class, and consequently consume more. \\

\noindent \textbf{Discussion.} When agents can choose their own friends, they end up grouping themselves with like-minded people. This creates a set of ``social classes'', where those with the highest benefit from network connections relative to the marginal cost of consumption: (i) form the strongest connections, (ii) make the strongest social comparisons, and (iii) consume the most. While I have suppressed heterogeneity in $c$ for clarity of exposition, this result goes through unchanged with heterogeneous $c_i$ (see Online Appendix). So the driving force is not the benefit of network connections alone. Rather it is the ratio of the benefit of connections and the cost of consumption.%It is up to the reader to decide whether they believe heterogeneity in connection benefits or in cost of consumption is the dominant force in driving social classes.
\footnote{One might argue that the parameter values might be correlated. It is plausible that agents with high values of $c_i$ might also have higher values of $b_i$, as difficulty in purchasing consumption increase the risk-sharing and insurance benefits of social networks. Any such correlation has no impact on the result.} The result, and its key insight -- that social comparisons drive people into distinct social classes -- does not depend on which factor dominates.

% Para 9
This is obviously a stylised benchmark -- we rarely have complete control over who we spend time with in the real world. Family, coworkers and the neighbours across the street are difficult to avoid. We should not expect to see such a stark result in practice. Rather, \Cref{prop: endogenous network} shows that, among the friends agents do freely choose, social comparisons push them to choose others in the same ``social class''. Naturally, they will have some links to other social classes, most plausibly through family or coworker ties. It is connections with coworkers that I turn to next.

%% file: 4.labour_v2.tex
% Para 1
One appealing interpretation for the cost parameter $c$ is as a cost of effort of working to acquire consumption. Implicitly, this fits a world where people work on their own and can only increase their consumption by working more or harder. In reality however, most people work for firms. And people can attain higher consumption by moving to more productive firms, without necessarily working harder. However, moving firms comes with an obvious catch -- your coworkers change. This is unavoidable. It is this change in the coworkers, and the accompanying changes in social comparisons with them, that constitutes the effective cost of changing firms.\footnote{In practice changing firms will often involve changing both coworkers and effort. I abstract from this, and focus on the role of social comparisons.}

A significant challenge to modelling these social comparisons when people are changing firms is that the set of coworkers changes entirely. It is not meaningful to separate changes in a person's firm from their network of coworkers. This is not just a technical challenge, but a conceptual one. When someone changes firms it is unclear who they will form links with in a new firm, and how the structure of these links might change as a result of the move. A natural benchmark is to assume that that people form equally strong links with all of their coworkers. Under this assumption an agent compares herself to average coworker consumption.\footnote{%This complete network assumption is made throughout \cite{ghiglino2010keeping}. 
Given this assumption, ``firm'' is best interpreted as a team or department. It is unreasonable to think that an agent compares herself to everyone in her organisation. Far more plausible is the claim that she compares herself to the more modest number of people she works with day-to-day.} 

The only difference between friends and coworkers is that an agent stops making social comparisons with her old coworkers when she moves firm, and starts making comparisons with her new ones. In Section \Cref{sec:results} I ignored the distinction between friends and coworkers because the distinction is mute when agents do not change firms. Note that in reality most of us keep in contact with, and therefore keep making social comparisons with, some people we used to work with. In my model, these people count as friends.\footnote{A more nuanced view recognised that when we stop working with someone we often spend less time with them, and so will make weaker social comparisons. My model copes with this easily. Such a person should considered both a friend and a coworker. The social comparisons that fall away when you stop working with them are the coworker component. The comparisons that remain are the friends component.} Here, we can split the network into the friends component and the coworkers component. It is merely conceptual, but nevertheless helpful in thinking about the problem. So $i$'s utility function is:
\begin{align}
    u_i = f(x_{i} - \alpha_{1i} \sum_{j \in \text{friends}} g_{ij} x_{j} - \alpha_{2i} \overline{x}_m) + b \sum_j \alpha_{1i} g_{ij}
\end{align}
where $\overline{x}_m$ is the average consumption of coworkers at $i$'s firm. Notice that this is identical to \Cref{eq: prefs} except there is now no direct marginal cost of consumption, and I have separated out friends and coworkers for clarity.

I assume that their are two types of agent, \emph{skilled} ($S$) and \emph{unskilled} ($U$), and two types of firm, \emph{high productivity} ($H$) and \emph{low productivity} ($L$). Conditional on their type, every agents earns, and therefore consumes, more when she works at a high productivity firm than a low productivity firm. Conditional on the firm's type, a skilled agent earns more than an unskilled agent. Formally: $x_{iH} > x_{iL}$, for all $i$, and $x_{Sm} > x_{Um}$ for $m \in \{H,L\}$.\footnote{I leave the process of wage-setting in the background as it is not the focus of this paper. Note however that I implicitly assume that an agent's wage does not depend on who else works at the firm. I am abstracting from notions of team production.}

% Para 6
Further, each firm has a fixed number of job openings. This is stark, and corresponds to a short-run world where firms face some friction that prevents them from either adjusting wages or creating/destroying jobs. I also assume that all firms prefer to hire skilled over unskilled workers. This is a consequence of pushing the wage-setting process into the background. For simplicity, I also assume that firms are large -- in the sense that they employ lots of workers. This is not an important assumption, but makes the the average coworker consumption, $\overline{x}_m$, much simpler.\footnote{This allows me to ignore the impact a single worker has on the social comparisons others at the firm.} %In particular, it means $\overline{x}_m = c x_{SH} + (1 - c) x_{UH}$ for firms $m \in H$. And similarly, $\overline{x}_m = (1 - c) x_{SL} + c x_{UL}$ for firms $m \in L$. 

I also assume that the total strength of social comparisons, $\alpha_{1i} + \alpha_{2i} \equiv \alpha_{i}$, is fixed for each agent and (to guarantee existence) that $\alpha_{i} < 1$. This builds in a trade-off between social comparisons with friends and with coworkers. One motivation for this is that an agent's propensity to compare herself with others is fixed outside of this model, by a set of underlying traits like her personality. Given that, the relative comparisons with different people depend on the time she spends with them. Spending more time comparing herself with friends necessarily entails less time comparing herself with coworkers. This assumption is a particularly parsimonious way of introducing such a trade-off. %This trade-off would exist if agents have separate reference points for friends and coworkers. I consider this variation in the Online Appendix. It achieves the same thing, but is less elegant. 

We can now examine how social comparisons determines agents' \emph{sorting} into firms. But first, we need a formal definition of sorting. With only two types of agents and two types of firms, there is a very simple definition.

\begin{defn}[Sorting]\label{defn:sorting}
	Sorting is equal to the fraction of skilled workers that work for high productivity firms, $c$.
\end{defn}

Sorting captures the extent to which workers and firms match along lines of skill and productivity. Absolute deviations of $c$ from $1/2$ capture \emph{segregation} -- the extent to which workers are concentrated in firms with other workers of the same type. If the consumption premium from working at a high productivity firm (compared to a low productivity firm) is larger for skilled workers than for unskilled ones, then skilled workers will never disproportionately concentrate in low productivity firms in equilibrium. So sorting and segregation coincide, and it suffices to simply consider $c$. This is because the only incentive for a skilled worker to choose a low productivity firm is to reduce her reference point by having lower-earning co-workers. If skilled workers worked disproportionately for low productivity firms, then this incentive would not be present.

The key insight in this section is that social comparisons with coworkers \emph{induce} less labour market sorting.

\begin{prop}\label{prop:sorting}
	If the strength of social comparisons with coworkers weakly increases for all workers, then labour market sorting weakly decreases.
\end{prop}

% Para 11
Social comparisons with coworkers can cause a skilled worker to choose to work for low productivity firms because they want to be a \emph{``big fish in a small pond''}. This is because average wages are lower at low productivity firms, and having lower coworker wages reduces her reference point. When comparisons with coworkers are strong enough, the benefit of improving her wage \emph{relative to her co-workers} outweighs the cost of a lower \emph{absolute} wage. If social comparisons with coworkers become stronger, then the ``big fish'' effect becomes more prominent, so skilled workers shift to low productivity firms. 

Labour market sorting in turn impacts a range of other outcomes. Therefore we can easily link social comparisons and these other outcomes. For example, if output and wages (which are implicitly functions of agent and firm types) are supermodular, then an increase in the strength of social comparisons with coworkers weakly reduces output and reduces the variance in wages.\footnote{The variance result does not extend to standard measure of inequality, such as the Gini coefficient, because the strength of social comparisons affect the average wage.} If the functions are submodular, the results are flipped.

%Notice that friends do not play a direct role here. This is because agents only form a single reference point from both friends and coworkers together. Agents are always willing to trade off favourable comparisons with friends for favourable comparisons with coworkers in a one-for-one manner. In the Online Appendix I show a variation where agents make social comparisons with friends and coworkers \emph{separately}. That is, they have two separate reference points. There, social comparisons with friends play do a direct role. 

%There, comparisons with friends work in the opposite direction to comparisons with coworkers. Consequently, friends matter for choosing a firm. Conditional on the strength of coworker comparisons, an agent who compares herself more strongly with friends is more inclined to work for the high productivity firm. When she cares more about her friends, the marginal benefit of getting ahead of them is greater, so she is more willing to fall behind her coworkers to achieve this.

Notice that the structure of the social network (of friends) does not play a direct role here. This is because it, and friends' wages, do not change in response to a worker's choice of firm. In contrast, the \emph{strength} of social comparisons with friends does matter, because of the assumption that total weight a workers places of social comparisons is fixed. This means that stronger comparisons with friends causes weaker comparisons with coworkers. Consequently, outcomes for skilled workers -- where they work and how much they earn -- are intimately tied to their social networks. 

\begin{prop}\label{prop:network sorting}
There exists a threshold value $\alpha_1^{crit}$ such that a skilled worker works at a high productivity firm if and only if $\alpha_{1i} \geq \alpha_1^{crit}$.
\end{prop}

This means that all skilled workers at high productivity firms make stronger social comparisons with friends than all skilled workers at low productivity firms. Those with the strongest social comparisons with friends experience the weakest desire to be a ``big fish in a small pond'', and so choose to work for high productivity firms. 

The result gives a complete characterisation of which type of firm every skilled worker chooses. A fraction $c$ of skilled workers work at high productivity firms. These are precisely the fraction $c$ of skilled workers with the strongest comparisons with friends. The result does not extend to unskilled workers because unskilled workers in my model are forced to take whatever vacancies are left over, and firms are indifferent between all unskilled workers.

\subsection{Discussion}
% Para 16: Evidence that comparisons matter for who works where
Conceptually, this section merely divides an agent's neighbours into coworkers and friends (that is, people she stops comparing herself with when she changes firm, and those who remain constant) and considers changes in coworkers as the cost of moving firms. However, this decomposition assumes that agents do actually make social comparisons with coworkers and are willing to change jobs as a result. There is ample empirical evidence to support this assumption. In field experiments both \cite{cohn2014social} and \cite{breza2018morale} find that creating social comparisons within teams reduces the effort exerted by the lower-paid team members. Obviously my model abstracts from effort, but these studies show that comparisons with coworkers are present and do affect behaviour. Closer to my model, in another field experiment \cite{card2012inequality} find that disclosing coworkers' pay makes relatively lower paid less satisfied with their job and more likely to search for a new one.

%\cite{cohn2014social} conduct a field experiment involving German sales workers, who are allocated into teams of two. They show that these workers prefer higher \emph{relative} income but lower \emph{absolute} income compared to the reverse (i.e. lower relative income but higher absolute income). They also find that social comparisons had a large impact on workers' effort. While my model abstracts from effort, \cite{cohn2014social} demonstrate that social comparisons have an important impact on behaviour in the labour market.  \cite{breza2018morale} make similar findings in a field experiment with Indian manufacturing workers, but use larger groups of workers. Related, \cite{card2012inequality} ... 

% Paper 4: Lab experiment
%Further, \cite{fu2019social} conduct an experiment of the McCall search model with the adjustment that job-seekers can see the payoffs of a partner. A job-seekers' reservation wage is increasing in the wage achieved by the partner. Again, this demonstrates that social comparisons do in induce agents to take different jobs.

The model can then explain two key stylised facts in labour market behaviour: that sorting and segregation in labour markets are incomplete and that they have risen significantly over recent decades \citep{song2019firming, card2013workplace}.\footnote{\cite{song2019firming} find the majority of the increase in inequality in the US over the past 40 years is due to this increased sorting and segregation. So it is obviously an important question on its own terms. \cite{song2019firming} use a confidential database from the U.S. Social Security Administration, which contains earnings records for every individual ever issued a U.S. Social Security number. Earnings data are uncapped and contain all sources of income deemed as remuneration by the U.S. Internal Revenue Service. They use the full dataset from 1981 to 2013.} These analyses raise two questions. Why didn't skilled workers sort into high productivity firms and only work with similarly-skilled peers in the first place? And what caused sorting and segregation to increase? In my model, the answer to the first question is straightforward -- social comparisons. Comparisons with coworkers make agents to want to be a ``big fish in a small pond'', causing incomplete sorting. 

An increase in sorting follows from agents making (relatively) weaker social comparisons with coworkers. I contend that this has been driven by changes in technology -- most importantly the rapid growth in social media use over recent decades.\footnote{In 2021, 72\% of Americans use at least one social media platform, up from only 5\% in 2005 [Pew Research, \citeyear{pew2021social}].} This works through two mechanisms. First, it has made it easier to observe looser friends' and acquaintances' consumption \citep{liu2016meta, ellison2007social}, and compare yourself to them \citep{verduyn2017social}. Second it has made it easier to maintain contact with \emph{ex}-coworkers. In my model this transforms coworkers in friends. Remember that the definition of friends is simply people who you continue comparing yourself to when you change firm.\footnote{There is absolutely no requirement that you actually like them.} In contrast, there has been no change of comparable magnitude to the ways in which we can compare ourselves to our coworkers.

There are of course a number of existing explanations for rising sorting. Some popular arguments include: rising return to skills (e.g. \cite{acemoglu2011skills}), increasing use of outsourcing, and increasing worker-firm complementarities. Rising returns to skills would mechanically increase sorting and segregation, but \citet[p38]{song2019firming} show that this can only explain a (small) minority of the changes. Increases in outsourcing are similarly an important component. \cite{goldschmidt2017rise} show it explains 9\% of the increase in German wage inequality since the 1980's. Again, this can only be one part of a bigger picture.

%Detailed empirical evidence regarding changes in worker-firm complementarities does not appear to be available. However, there are a variety of models...
%ADD PARAS WORKER-FIRM COMPLEMENTARITY. \\

An important question is whether my model has different implications to existing explanations. In a sense, I simply provide another mechanism linking changes in technology to increases in labour market sorting. The precise technology change -- the rise of social media -- is different to other explanations, but this has happened at the same time as many other changes that stem from increasing internet use. So how can we distinguish my mechanism from existing one?

% Para 21:
My model provides clear predictions -- albeit under stark assumptions -- about \emph{who works where}, in addition to to just the \emph{extent} of sorting. Skilled workers choose high wage jobs because the pressure from their social network to earn more outweighs their desire to be the ``big fish'' in the workplace. So skilled workers with weaker social networks are the ones who choose low productivity firms. This is because they feel the ``big fish'' pressure (relatively) more strongly. It is the particular pattern of sorting -- the link between an individual's strength of social comparisons and where she chooses to work -- that separates my model from existing explanations.

%While there is little detailed evidence linking changes in social comparisons to increasing labour market sorting, there is a variety of anecdotal evidence.

There is anecdotal evidence that suggests people with weaker social networks -- in the sense of their friends having lower consumption -- end up with lower paying jobs themselves. See for example \cite{jackson2019human} and \cite{putnam2016our}. This is in line with the predictions of the my model.\footnote{Note however, that it is also consistent with a complementary explanation -- that social networks affect agents' \emph{ability} to get high paying jobs.} In my set-up, it is something akin to a poverty of aspiration that pushes some skilled workers into  lower productivity, lower paying, firms. They are not held back by a lack of ability or insider access. Rather, it is their preferences that drive them to earning lower wages.\footnote{The implications for policy-makers are mixed. On the one hand, giving (currently) poorer skilled agents stronger social networks can increase their earnings. However, this requires changing these agents' preferences, and leaves them with lower utility.} 

% Para 23:
This link between social networks, preferences, and in turn, labour market decisions, is in line with existing work. Existing empirical work shows that an agent's social connections can affect participation \citep{nicoletti2018family}), productivity \citep{mas2009peers}, and enrolment in retirement plans \citep{duflo2003role}, among others. My result suggests a new margin -- namely the choice of where to work.

%% file: 5.conclusion.tex
This paper has taken seriously the simple idea that keeping up with ``The Joneses'' matters. Doing so has a number of striking implications for policymakers. Stronger social comparisons increase consumption but reduce welfare. The usual close relationship between consumption and welfare need not exist. This suggests standard consumption-based proxies for welfare, such as G.D.P., may be a poor guide when targeting welfare. It complements work related to the Easterlin puzzle \citep{easterlin1995will, frey2002can}, which asks why happiness correlates with incomes at a point in time, but doesn't continually increase as incomes rise. If we believe that social comparisons are a significant part of agents' preferences, then we should be much more careful in advocating for continued economic growth (\cite{easterlin2013happiness} makes a similar argument using macro-level data).

Further, increasing the cost of consumption can benefit everybody -- taxation can be welfare improving, even if the government wastes the revenue. So goods where social comparisons are important are ideal candidates for taxation when a government does need to raise revenue. This complements the standard advice to target goods with negative externalities and markets where taxes do not distort behaviour much. When the network is endogenous -- when agents can choose links for themselves -- a clear set of ``social classes'' emerge. Agents only form friendships with people like them. This happens because everyone wants to escape the pressures of making social comparisons with those richer than themselves, and as a consequence they can only become friends with those at the same level of consumption.

Extending the model to examine a simple labour market helps explain two stylised facts from empirical studies of labour markets. First, labour market sorting is incomplete -- some skilled workers continue to work for low productivity firms. Second, the degree of sorting has been rising over time. Social comparisons with co-workers induce incomplete sorting because they make agents want to be a ``big fish a small pond''. At a certain point they are willing to move to a low productivity firm, losing income in absolute terms, in exchange for becoming better off relative to their (new) co-workers. I contend that the rise of social media has made it much easier to compare yourself to your friends, and so has made co-workers a relatively less important source of social comparisons. This trend would drive the increase in labour market sorting found in the data \citep{song2019firming}.

Social comparisons can also determine which workers work for which firms, at least among skilled workers. Skilled workers with weaker social networks will choose to work for low productivity firms. Without the social pressures from their friends to earn high incomes, the ``big fish in a small pond'' effect looms relatively large for them. This suggests that social networks affect labour market outcomes through impacts on agents' \emph{preferences}, rather than just opportunities.

%This paper set out a model of reference dependent choice, where agents' reference points are determined from social comparisons. This setup generates a network game of strategic complements, similar to \cite{ballester2006s}. In line with this, we find a unique equilibrium where consumption depends critically on Bonacich centrality. However, our setup provides a unique measure of Bonacich centrality, rather than a family of measures that depend on a decay parameter.

%% file: A1.proofs_consumption.tex
Before presenting the proofs, I make two simplifications to notation. (1) Let $B = (I - G)^{-1}$, and $B_{ij} \equiv (I - G)^{-1}_{ij}$, (2) $F(\cdot) \equiv f^{'-1}(\cdot)$. We only use these where they aid exposition. %We begin by stating a property of the function $f(\cdot)$.

\begin{proof}[\emph{\textbf{Proof of Remark \ref{prop:equilibrium}}}]
%The First Order Condition is: $x_i^* - \alpha_i \sum_i g_{ij} x_j^* = f^{'-1}(c)$. 
In matrix notation, the First Order Condition is: $x^* (I - G) = f^{'-1}(c)$. $\lambda_1 < 1$ implies that $(I - G)^{-1}$ exists, is unique, and can be represented by the Neumann series $(I - G)^{-1} = \sum_{k=0}^{\infty} G^k$. 
\end{proof}
% ------------------------------------------------------------- %

\noindent At this point, it is helpful to state a relationship that holds in equilibrium. It is obvious (so a proof is omitted), but I will use it in various proofs.

% ------------------------------------------------------------- %
% LEMMA: STATES A PAIR OF HELPFUL RELATIONSHIPS FROM THE F.O.C.
% ------------------------------------------------------------- %
\begin{lem}\label{lem:foc in eq}
In equilibrium, $f'(x_i^* - \alpha_i \sum_i g_{ij} x_j^*) = c$ and $x_i^* - \alpha_i \sum_i g_{ij} x_j^* = f^{'-1}(c)$
\end{lem}
% ------------------------------------------------------------- %
% PROVES THE LEMMA IMMEDIATELY ABOVE
% ------------------------------------------------------------- %
%\begin{proof}[\emph{\textbf{Proof of Lemma \ref{lem:foc in eq}}}]
%This follows from the fact that every agent's First Order Condition must be satisfied in equilibrium. %The first expression is found by substituting equilibrium consumption back into the First Order Condition, and the second by rearranging.
%\end{proof}
% ------------------------------------------------------------- %

% ------------------------------------------------------------- %
% PROVES THE COMPARATIVE STATICS (CONSUMPTION AND WELFARE) FOR REFERENCE STRENGTH
% ------------------------------------------------------------- %
\begin{proof}[\emph{\textbf{Proof of Proposition \ref{prop:ref strength}}}]
\textbf{(i)} $(I - \text{diag}(\alpha) g)^{-1} = \sum_{k=0}^{\infty} (\text{diag}(\alpha) g)^{k}$. This is the a standard Neumann Series (see for example \cite{horn2012matrix}). If a given $\alpha_{j}$ increases, then all elements of $(\text{diag}(\alpha) g)^{k}$ weakly increase for any $k$. Since $x_{i}^{*} \propto \sum_{j} (I - \text{diag}(\alpha) g)^{-1}_{ij}$, then $x^{*}_{i}$ is weakly increasing for all $i$. This relationship must be strict when $i=j$ since $( \text{diag}(\alpha) g)_{ii}$ is strictly increasing in $\alpha_{i}$ by construction.

\noindent \textbf{(ii)} Take equilibrium utility, $u_i^* = f(x_i^* - \alpha_i \sum_i g_{ij} x_j^*) - c x_i^*$, and differentiate with respect to $\alpha_k$.
\begin{align*}
    %u_i^* &= f(x_i^* - \alpha_i \sum_i g_{ij} x_j^*) - c x_i^* \\
    \frac{d u_i^*}{d \alpha_k} &= \left( \frac{d x_i^*}{d \alpha_k} - \alpha_i \sum_i g_{ij}  \frac{d x_i^*}{d \alpha_k} \right) f'(x_i^* - \alpha_i \sum_i g_{ij} x_j^*) - c \frac{d x_i^*}{d \alpha_k}
\end{align*}
Applying \Cref{lem:foc in eq} an simplifying yields: $\frac{d u_i^*}{d \alpha_k} = - \alpha_i c \sum_i g_{ij}  \frac{d x_i^*}{d \alpha_k}$
%\begin{align*}
%    \frac{d u_i^*}{d \alpha_k} = - \alpha_i c \sum_i g_{ij}  \frac{d x_i^*}{d \alpha_k}
%\end{align*}
Finally, noticing that $x_i^*$ is weakly decreasing in $\alpha_k$ for all $i,k$ (and that $c>0, \alpha \geq 0, g_{ij} \geq 0$) yields the result.
\end{proof}
% ------------------------------------------------------------- %

% ------------------------------------------------------------- %
% PROVES THE COMPARATIVE STATICS (CONSUMPTION AND WELFARE) FOR COST
% ------------------------------------------------------------- %
\begin{proof}[\emph{\textbf{Proof of Proposition \ref{prop:cost}}}]
\textbf{(i)} Follows from the equilibrium found in Proposition \ref{prop:equilibrium} and the fact that $f^{'-1}(\cdot) \equiv F(c)$ is strictly decreasing and convex.
\textbf{(ii)} First, recall that $x_i^* = C_i^b F(c)$, and that $C_i^b$ depends only on the network, and not on $c$. Therefore $\frac{d x_i^*}{d c} = F'(c) C_i^b$ Now take equilibrium utility, $u_i^* = f(x_i^* - \alpha_i \sum_i g_{ij} x_j^*) - c x_i^*$, and differentiate with respect to $c$.
\begin{align*}
    \frac{d u_i^*}{d c} = F'(c) (C_i^b - \alpha_i \sum_i g_{ij} C_j^b) f'(x_i^* - \alpha_i \sum_i g_{ij} x_j^*) - c F'(c) C_i^b - F(c) C_i^b
\end{align*}
Notice that $F(c) (C_i^b - \alpha_i \sum_i g_{ij} C_j^b) = x_i^* - \alpha_i \sum_i g_{ij} x_j^*$, which is equal to $F(c)$ by \Cref{lem:foc in eq}. Therefore $(C_i^b - \alpha_i \sum_i g_{ij} C_j^b) = 1$. Applying \Cref{lem:foc in eq} again and simplifying yields:
\begin{align*}
    \frac{d u_i^*}{d c} = F'(c) c - c F'(c) C_i^b - F(c) C_i^b
\end{align*}
By rearranging, we can see that $\frac{d u_i^*}{d c} > 0$ if and only if: $ \frac{C_i^b - 1}{C_i^b} > \frac{F(c)}{- F'(c) \cdot c}$ (during this rearranging the inequality sign will flip when we divide through by $(1 - C_i^b)$ because this term is negative).
\end{proof}

% COMMENT: WE NEED THE PERTURBATION THEOREM TO SOLVE ALL THE NETWORK PERTURBATION RESULTS
\noindent The following Lemma restates the results of \cite{chang2006inversion}. It underpins all results relating to comparison shifts of the network. To my knowledge, this result is not well known in economics, so I provide a proof in \Cref{extension:chang2006}.

% ------------------------------------------------------------- %
% LEMMA: CHARACTERISES THE EXACT IMPACT OF AN ELEMENTARY PERTURBATION
% ------------------------------------------------------------- %
\begin{lem}[A perturbation theorem. \citet{chang2006inversion}]\label{perturbation theorem}
If $D$ is a comparison shift, then $(I - [G + D])^{-1} - (I - G)^{-1} = H$. Where:
%\begin{align*}
%H = \frac{\phi}{1 - \phi (B_{ur} - B_{dr}) }
%	\begin{pmatrix}
%	B_{1r} (B_{u1} - B_{d1}) & B_{1r} (B_{u2} - B_{d2}) & \cdots & B_{1r} (B_{un} - B_{dn}) \\
%	B_{2r} (B_{u1} - B_{d1}) & B_{2r} (B_{u2} - B_{d2}) & \cdots & B_{2r} (B_{un} - B_{dn}) \\
%	\vdots & \vdots & \ddots & \vdots \\
%	B_{nr} (B_{u1} - B_{d1}) & B_{nr} (B_{u2} - B_{d2}) & \cdots & B_{nr} (B_{un} - B_{dn})
%	\end{pmatrix}
%\end{align*}
% SHOW A SMALLER VERSION
\begin{align*}
H = \frac{\phi}{1 - \phi (B_{ur} - B_{dr}) }
	\begin{pmatrix}
	B_{1r} (B_{u1} - B_{d1}) & \cdots & B_{1r} (B_{un} - B_{dn}) \\
	\vdots &  \ddots & \vdots \\
	B_{nr} (B_{u1} - B_{d1}) & \cdots & B_{nr} (B_{un} - B_{dn})
	\end{pmatrix}
\end{align*}
\end{lem}
% ------------------------------------------------------------- %
% PROOF SHOWN IN A SEPARATE APPENDIX.

% ------------------------------------------------------------- %
% LEMMA SHOWING THE EXACT EFFECT OF AN ELEMENTARY PERTURBATION FOR ONE AGENT
% ------------------------------------------------------------- %
\begin{lem}\label{lem:perturb effect}
	The change in agent $i$'s consumption due a comparison shift $D$, $\Delta x_i^*$, is equal to:
	$$\frac{ \phi B_{ir} (C^{b}_{u} - C^{b}_{d}) }{1 - \phi (B_{ur} - B_{dr}) } \cdot f'^{-1}(c)$$
\end{lem}
% ------------------------------------------------------------- %
% PROVES THE LEMMA IMMEDIATELY ABOVE
% ------------------------------------------------------------- %
\begin{proof}[\emph{\textbf{Proof of Lemma \ref{lem:perturb effect}.}}]
From Proposition~\ref{prop:equilibrium}: $x_{i}^{*} %= f'^{-1}(c) \sum_{j} (I- G)^{-1}_{ij}
= f'^{-1}(c) \sum_{j} B_{ij}$ (by the definition of $B$). From Lemma~\ref{perturbation theorem}, the optimal action after a comparison shift $D$ is: $x_{i}^{'*} = f'^{-1}(c) \sum_{j} (B_{ij} + H_{ij})$. Denote the change in optimal action $x_{i}^{'*} - x_{i}^{*} \equiv \Delta x_{i}^{*}$. Then using the definition of $H$: $\Delta x_{i}^{*} = f'^{-1}(c) \sum_{j} \frac{ \phi B_{ir} (B_{uj} - B_{dj}) }{1 - \phi (B_{ur} - B_{dr}) }$. Finally, using the definition of Bonacich centrality (\Cref{defn:bonacich centrality}) yields the result.
%$$\Delta x_{i}^{*} = f'^{-1}(c) \cdot \frac{ \phi B_{ir} (C^{b}_{u} - C^{b}_{d}) }{1 - \phi (B_{ur} - B_{dr}) }$$
\end{proof}
% ------------------------------------------------------------- %

% ------------------------------------------------------------- %
% PROVES THE SIZE AND CONVEXITY/CONCAVITY OF EFFECT OF AN ELEMENTARY PERTURBATION
% ------------------------------------------------------------- %
\begin{proof}[\emph{\textbf{Proof of Proposition \ref{prop:perturbation}}}]
\textbf{(i)} Arbitrarily small changes in $\phi$ are possible, so $\Delta x_{i}^{*}$ is continuously differentiable in $\phi$. For ease of exposition, let: $ B_{ir} (C^{b}_{u} - C^{b}_{d}) f'^{-1}(c) \equiv y$ and $(B_{ur} - B_{dr}) \equiv z$. Let $\hat{x}_i^* \equiv x_i^* + \Delta x_i^*$ be the new equilibrium value of consumption.
Then $\hat{x}_i^* = x_i^* + y \phi (1 - z \phi)^{-1}$. By standard rules of differentiation (and some rearranging): $\frac{d \hat{x}_i^* }{d \phi} = y \cdot (1 - z \phi)^{-2}$. $y > 0$ as $B_{ir} > 0$, $f'^{-1}(c) > 0$ (by construction), and $C^{b}_{u} - C^{b}_{d} > 0$ (by assumption). Therefore $\frac{d \hat{x}_i^* }{d \phi} > 0$. 
% WORKING OUT: \Delta x_{i}^{*} = y \phi (1 - z \phi)^{-1} --> \frac{d (\Delta x_{i}^{*}) }{d \phi} = y(1 - \phi z)^{-1} + y \phi (-1)(-1) z (1 - \phi z)^{-2} --> \frac{d (\Delta x_{i}^{*}) }{d \phi} = y(1 - \phi z)^{-2} [(1 - \phi z) + \phi z].
%
\textbf{(ii)} taking the second expression from \Cref{lem:foc in eq}, and substituting it into equilibrium utility, we have that $\hat{u}_i^* = f(F(c)) - c \hat{x}_i^*$. Since $\hat{x}_i^*$ is increasing in $\phi$, then $\hat{u}_i^*$ is clearly decreasing in $\phi$.
\end{proof}
% ------------------------------------------------------------- %

\noindent It is now helpful to state an implication of \Cref{lem:foc in eq}.

% ------------------------------------------------------------- %
% SHOWS THE PARTIAL DERIVATIVE WITH RESPECT TO A LINK (IN A FIXED NETWORK)
% ------------------------------------------------------------- %
\begin{lem}\label{lem:comp static u}
    For a fixed network $G$, $\frac{d u_i^*}{d G_{ij}} = - x_j^* c_i + b_i$
\end{lem}
% ------------------------------------------------------------- %
\begin{proof}[\emph{\textbf{Proof of Lemma \ref{lem:comp static u}.}}]
Take \cref{eq: prefs}, substitute in equilibrium values of consumption, $x_i^*$, and utility, $u_i^*$, and then differentiate with respect to $G_{ij}$. This yields:
\begin{align*}
    \frac{d u_i^*}{d G_{ij}} &= \left( \frac{d x_i^*}{d G_{ij}} - x_j^* \right) f' \left( x_i^* - \sum_j G_{ij} x_j^* \right) - c_i \frac{d x_i^*}{d G_{ij}} + b_i
\end{align*}
From \Cref{lem:foc in eq} we know that $f'(x_i^* - \sum_j G_{ij} x_j^*) = c_i$. Substituting this in and simplifying yields the result.
\end{proof}
%For a fixed network, $G$, an agent $i$ has  equilibrium consumption [resp. utility] of $x_i^*$ [resp. $u_i^*$]. Therefore:
%\begin{align*}
%    u_i^* &= f \left( x_i^* - \sum_j G_{ij} x_j^* \right) - c_i x_i^* + b_i \sum_j G_{ij} \\
%    \frac{d u_i^*}{d G_{ij}} &= \left( \frac{d x_i^*}{d G_{ij}} - x_j^* \right) f' \left( x_i^* - \sum_j G_{ij} x_j^* \right) - c_i \frac{d x_i^*}{d G_{ij}} + b_i
%\end{align*}
%We know from \Cref{lem:foc in eq} that $f'(x_i^* - \sum_j G_{ij} x_j^*) = c_i$. Substituting this in and simplifying yields:
%\vspace{-8mm}
%\begin{align*}
%    \frac{d u_i^*}{d G_{ij}} &= - x_j^* c_i + b_i
%\end{align*}
\noindent It is also clearer to prove \Cref{cor:endogenous network} and \Cref{prop: endogenous network} together.
% ------------------------------------------------------------- %

% ------------------------------------------------------------- %
% PROVES THE ENDOGENOUS NETWORK FORMATION RESULT
% ------------------------------------------------------------- %
\begin{proof}[\emph{\textbf{Proof of Proposition \ref{prop: endogenous network} and Corollary \ref{cor:endogenous network}}}]
\textbf{(i)} First, notice that when $i$ has no links, $x_i^* = f^{'-1}(c)$. Therefore, $b_i < c f^{'-1}(c)$ then $i$ will never want to form a link (as no agent can have consumption below $f^{'-1}(c)$).
\textbf{(ii)} Now consider cases where $b_i \geq c f^{'-1}(c)$. Consider any pair of agents $i,j$ such that $b_i=b_j$. Recall that in \cref{sec:model} we \emph{assumed} that for any $i$, there exists some $j$ such that $b_i=b_j$. Since $\frac{d u_i^*}{d G_{ij}} = - x_j^* c_i + b_i$ for all agents, then $i$ and $j$ want to form (or strengthen) a link if $x_i^* < \frac{b_i}{c}$ and $x_j^* < \frac{b_i}{c}$. Therefore, $x_i^* \geq \frac{b_i}{c}$ for all $i$.
\textbf{(iii)} Now notice that $i$ does not want to form a link with any agent $k$ where $b_k > b_i$ (or wants to cut a link if one exists). This is because $x_k^* \geq \frac{b_k}{c}$ which implies that $-x_k^* c + b_i < 0$. Since links require mutual consent, there can be no links between two agents $i,j$ where $b_i \neq b_j$. This is because either $b_i > b_j$ or $b_j > b_i$, so one of the agents would want to cut the link.
\textbf{(iv)} Since agents can only form links with others who share a value of $b$, then $x_i^* \leq \frac{b_i}{c}$. Otherwise, $-x_i^* c + b_j < 0$ (since $b_i = b_j$) and so $j$ would cut the link. Therefore $x_i = \frac{b_i}{c}$ for all $i$ where $b_i < c f^{'-1}(c)$.
\textbf{(v)} We know from \Cref{prop:equilibrium} that for any network $G$, $C_i^b = \frac{x_i^*}{f^{'-1}(c)}$. The Bonacich centrality result in \cref{cor:endogenous network} follows trivially from substituting the value $x_i^*$ found above into this equation.
\end{proof}
% ------------------------------------------------------------- %

%% file: A2.proofs_labour.tex
% ------------------------------------------------------------- %
% PROVES THAT STRONGER SOCIAL COMPARISONS WEAKLY DECREASE SORTING
% ------------------------------------------------------------- %
\begin{proof}[\emph{\textbf{Proof of \Cref{prop:sorting}.}}]
	\textbf{Firm choice.} Firms always choose skilled workers over unskilled workers (by assumption), so it suffices to consider where the skilled workers choose to work. Unskilled workers fill up all remaining vacancies (since we have assumed that there are as many vacancies as workers). A skilled worker $i$ will apply for a job at a high productivity firm only if and only if:
	\begin{align}
		f(x_{SH} - \alpha_{1i} \sum_{l} g_{ij} \overline{x}_{j} - \alpha_{2i} \overline{x}_{H}) &\geq f(x_{SL} - \alpha_{1i} \sum_{l} g_{ij} x_{j} - \alpha_{2i} \overline{x}_{L})
	\end{align}
	Simplifying and rearranging yields $x_{SH} - \alpha_{2i} \overline{x}_{H} \geq x_{SL} - \alpha_{2i} \overline{x}_{L}$. After further rearranging: $1/\alpha_{2i} \geq (\overline{x}_{H} - \overline{x}_{L}) / (x_{SH} - x_{SL})$. Let $c$ be the fraction of skilled workers working for at high productivity firms. Then $\overline{x}_{H} = c x_{SH} + (1 - c) x_{UH}$. Using this and the corresponding expression for $\overline{x}_{L}$ yields: 
	\begin{align}\label{eq:firm choice ineq}
	    \frac{1}{\alpha_{2i}}  \geq \frac{[(x_{SH} - x_{UH}) + (x_{SL} - x_{UL})] \cdot c + (x_{UH} - x_{UL}) }{ (x_{SH} - x_{SL}) } .
	\end{align}
	$(x_{SH} - x_{UH})$ [resp. $(x_{SL} - x_{UL})$]  is the skill premium at the high [resp. low] productivity firm. $(x_{SH} - x_{SL})$ [resp. $(x_{UH} - x_{UL})$] is the ``firm premium'' for a skilled [resp. unskilled] worker. By assumption, all these terms are strictly positive, so RHS of \Cref{eq:firm choice ineq} is strictly increasing in $c$. Clearly the LHS of \Cref{eq:firm choice ineq} is strictly decreasing in $\alpha_{2i}$.
	For any $c$ there exists a critical value of $\alpha_{2i}$ such that equation \ref{eq:firm choice ineq} holds with equality.\footnote{i.e. skilled worker $i$ is indifferent between working at a high productivity and low productivity firm} Denote this value $\alpha_{2i}^{crit}$. So $\alpha_{2i}^{crit}$ is a strictly decreasing function of $c$. Denote this function $\alpha_{2i}^{crit} = h(c)$, with $h'(c) < 0$.

	\textbf{Existence.} Let $F(\cdot)$ denote the C.D.F. of $\alpha_{2i}$ (i.e. the fraction of skilled workers with an $\alpha_{2i}$ value of less than or equal to $\alpha_2$). Therefore $F(h(c))$ is the fraction of skilled workers who want to work at high productivity firms, given that a fraction $c$ of skilled workers actually work at high productivity firms. An equilibrium is then a fixed point of this equation:  $F(h(c)) = c$. RHS of this is strictly increasing in $c$, and LHS is weakly decreasing in $c$, so there exists a fixed point, $c^*$.

	\textbf{Monotonicity.} Suppose that $\alpha_{2i}$ weakly increases for all $i$. Then $F(a)$ weakly decreases for all $a$ (this follows from well-known properties of C.D.F.s). This in turn reduces $F(h(c))$ for any $h(c)$, which weakly reduces $c^*$. Notice that $F(h(c)) = 1$ for all values of $c$ where $h(c) \geq \max_{i} \alpha_{2i}$.
\end{proof}
% ------------------------------------------------------------- %

% ------------------------------------------------------------- %
% PROVES THAT SKILLED WORKERS WITH HIGH DEGREE CENTRALITY WORK AT HIGH PRODUCTIVITY FIRMS
% ------------------------------------------------------------- %
\begin{proof}[\emph{\textbf{Proof of \Cref{prop:network sorting}.}}]
	From \Cref{prop:sorting}, for any $c$, there exists an agent-specific critical value $\alpha_{2i}^{crit}$ such that: \textbf{(a)} when $\alpha_{2i} = \alpha_{2i}^{crit}$, $i$ is \emph{indifferent} between high productivity and low productivity firms, \textbf{(b)} when $\alpha_{2i} < \alpha_{2i}^{crit}$, $i$ strictly prefers a \emph{high} productivity firm, \textbf{(c)} when $\alpha_{2i} > \alpha_{2i}^{crit}$, $i$ strictly prefers a \emph{low} productivity firm. \Cref{prop:sorting} also proved that there exists some equilibrium level of labour market sorting, $c^*$. 
	
	This yields an equilibrium critical value for $\alpha_{2}$. Denote it $\alpha_{2,eq}^{crit}$. Therefore, all skilled workers with $\alpha_{2i} \leq \alpha_{2,eq}^{crit}$ work at high productivity firms. All skilled workers with  $\alpha_{2i} \geq \alpha_{2,eq}^{crit}$ work at low productivity firms. 
	Finally, notice that $\alpha = \alpha_{1i} + \alpha_{2i}$ (by assumption), so a threshold in $\alpha_2$ is easily converted to a threshold in $\alpha_1$. %and that $\alpha_{1i}$ captures the strength of $i$'s comparisons with her friends (by definition). Therefore higher $\alpha_{1i}$ implies lower $\alpha_{2i}$.
\end{proof}
% ------------------------------------------------------------- %

%% file: B.no_network_effect.tex
Consider a special case where agents' reference strengths are uncorrelated with the structure of the social network. Its key feature is that the structure of the social network will not have any impact on outcomes. 

\begin{ass}[Uncorrelated case]\label{ass: mean field}
The reference structure, $g$, is uncorrelated with the reference strength: $corr(g_{ij} , \alpha_{j}) = 0$ for all agents $i,j$.
\end{ass}

While \Cref{ass: mean field} is appears rather abstract, it covers the natural case where all agents have the same reference strength ($\alpha_i = \alpha$ for all $i$), so it is not vacuous. Under this assumption, the network effects drop out.

\begin{cor}\label{cor:mean field no network}
Under Assumption \ref{ass: mean field}: $x_i^*$ and $u_i^*$ do not depend on the reference structure, $g$.
\end{cor}

% Intuition
\Cref{ass: mean field} requires that $i$'s neighbours are, on average, ``Mr and Mrs Average'' in terms of how strongly they make social comparisons. In turn, their neighbours are also, on average, ``Mr and Mrs Average'', and so on throughout the whole social network. There is no way for a global network effect to build -- it will always get averaged out. The clearest case where this assumption holds is where all agents have the same reference strength. This shows how homogeneity of the local social comparisons shuts down the effect of heterogeneity in global network structure.

% Extra results
This simplification also allows us to make stronger claims about the impact of reference strengths. First, an equilibrium exists if and only if the average reference strength is less than one, and Bonacich centralities collapse to a simple function of the reference strengths (in particular, $C_i^b = 1 + \frac{ \alpha_{i} }{1 - \overline{\alpha} }$ for all $i$). 

The tight link between an agent's reference strength and their optimal consumption makes it straightforward to see how changes in the distribution of these reference strengths feeds through into the distribution of consumption. First [resp. second] order stochastically dominant shifts in the distribution of the the references strengths causes a first [resp. second] order stochastically dominant shift in the distribution of equilibrium consumption.\footnote{First Order Stochastic Dominance and Second Order Stochastic Dominance formalise the notions of (unambiguously) ``bigger'' and ``more spread out'' respectively. They are first due to \cite{hadar1969rules} and \cite{rothschild1970increasing}. More modern coverage can be found in \cite[Ch 6.D]{mas1995microeconomic}. I prove the claim regarding equilibrium outcomes below. The additional claims follow straightforwardly from this.} 
Inequality in social comparisons drives consumption inequality, even in a setting where agents are otherwise identical. \\

% Policy discussion
\noindent \textbf{Implications for policy-making.} Measuring networks in practice -- especially the weighted networks present in my model -- is typically difficult and resource intensive. In a perfect world with unlimited resources and attention, we should of course examine whether this assumption holds in a network, and proceed accordingly. In reality, however, \cref{cor:mean field no network} presents a pragmatic approach for policy-makers and governments short on time and money, so long as the assumption is not violated too egregiously. \\

\subsection*{Proofs}
% ------------------------------------------------------------- %
% PROOF THAT NETWORK DOES NOT MATTER UNDER THE MEAN FIELD ASSUMPTION
% ------------------------------------------------------------- %
\begin{proof}[\emph{\textbf{Proof of Corollary \ref{cor:mean field no network}}}]
First, let $a \equiv \text{diag}(\alpha)$ for convenience. Using a Neumann Series representation, we can express $(I - ag)^{-1} = \sum_{k=1}^{\infty} (ag)^{k}$. Therefore, $C_i^b = \sum_j (I - ag)^{-1}_{j} = \sum_j \sum_{k=1}^{\infty} [(ag)^{k}]_{ij}$. We deal with the first two terms individually, and then all further terms by induction. Clearly $\sum_j \mathbf{1} \{ i = j \} = 1$, and $\sum_{j} (ag)_{ij} = \alpha_{i} \sum_{j} g_{ij} = \alpha_{i}$ by construction. Induction: for $k=2$;
\begin{align*}
\sum_{j} (ag)^{2}_{ij} &= \sum_{j} \sum_{s} (ag)_{is} (ag)_{sj}  \
		= \sum_{j} \sum_{s} \alpha_{i} g_{is} \alpha_{s} g_{sj}  \
		= \alpha_{i} \sum_{s} g_{is} \alpha_{s} \sum_{j} g_{sj}
\end{align*}
By Assumption \ref{ass: mean field}: $\sum_{s} g_{is} \alpha_{s} = \overline{\alpha} \sum_{s} g_{is}$. Then notice that $\sum_{j} g_{sj} = 1$ for any $s$. This yields: $\sum_{j} (ag)^{2}_{ij}	= \alpha_{i} \cdot \overline{\alpha}$.
Now assume that for $k=t > 2$, $\sum_{j} (ag)^{t}_{ij} = \alpha_{i} \cdot \overline{\alpha}^{ t-1}$. Then show for $k=t+1$. By definition $\sum_{j} G^{t+1}_{ij} = \sum_{j} [(ag) \cdot (ag)^{t}]_{ij}$. For clarity of exposition, we let $G^{t}_{ij} = \mathcal{A}_{ij}$.
\begin{align*}
\sum_{j} (ag)^{t+1}_{ij} = \sum_{j} \sum_{s} \alpha_{i} g_{is} \mathcal{A}_{sj}  \
		= \alpha_{i} \sum_{s} g_{is} \sum_{j} \mathcal{A}_{sj}
\end{align*}
Using the assumption for $k=t$, $\sum_{j} \mathcal{A}_{sj} = \alpha_{s} \overline{\alpha}^{t-1}$. Therefore, $\sum_{j} (ag)^{t+1}_{ij} = \alpha_{i} \cdot \overline{\alpha}^{t-1} \sum_{s} g_{is} \alpha_{s}$. Using Assumption \ref{ass: mean field} and the fact that rows of $g$ sum to $1$ as before:
$\sum_{j} (ag)^{t+1}_{ij} = \alpha_{i} \cdot \overline{\alpha}^{t}$. 
Therefore, $C_i^b = \sum_{j} \sum_{k=0}^{\infty} [(ag)^{k}]_{ij} = 1 + \alpha_{i} ( 1 + \overline{\alpha} + \overline{\alpha}^{2} + \overline{\alpha}^{3} + ...)$. Clearly this does not depend on the network. Finally, it follows from \Cref{prop:equilibrium} that $x_i^*$ and $u_i^*$ do not depend on the network structure, $g$.
\end{proof}
% ------------------------------------------------------------- %

%\Cref{ass: mean field} also allows us to make stronger statements about the effect of reference strengths. First, we can express the existence condition and Bonacich centralities (and by implication the equilibrium behaviour) solely in terms of the reference strengths.
\begin{rem}\label{rem:mean field existence}
Under Assumption \ref{ass: mean field}: (i) an equilibrium exists if and only if $\overline{\alpha} < 1$, and (ii) when $\overline{\alpha} < 1$: $C_i^b = 1 + \frac{ \alpha_{i} }{1 - \overline{\alpha} }$ for all $i$, where $\overline{\alpha} = \frac{1}{n} \sum_i \alpha_i$
\end{rem}
%Obviously these are identical definitions, but in this special case become much easier to interpret. Existence now depends on the reference strengths alone, rather than on an eigenvalue of the whole network. Centrality in the network collapses down to an agent's own reference strength, relative to the average in society.
% ------------------------------------------------------------- %
% PROOF OF STRONGER EXISTENCE RESULT UNDER THE MEAN FIELD ASSUMPTION
% ------------------------------------------------------------- %
\begin{proof}[\emph{\textbf{Proof of Remark \ref{rem:mean field existence}}}]
From the proof to Corollary \ref{cor:mean field no network}, we have $C_i^b = 1 + \alpha_i / (1 - \overline{\alpha})$. An equilibrium exists if and only if Bonacich centrality, $C_i^b$ is well defined. Clearly this happens if and only if $\overline{\alpha} < 1$.
\end{proof}
% ------------------------------------------------------------- %

%% file: C.nonlinear_subutility.tex
The model in \Cref{sec:model} assumes that the sub-utility is linear. This was a simplification compared to the \cite{koszegi2006model} benchmark. However, this was only for clarity of exposition. It is easier to understand how the model works when more things are linear. It also allows the proofs to use simpler machinery that readers, especially those somewhat familiar wit networks, ought to be more comfortable with.

Here, I introduce non-linear sub-utility. The model is exactly the same as in \Cref{sec:model}, except that $i$'s utility function is now:
\begin{align*}
u_{i} = f \left( m(x_{i}) - \alpha_{i} \sum_{j} g_{ij} m(x_{j}) \right) - c x_{i} + b_i \sum_j \alpha_i g_{ij}  \tag{1*}
\end{align*}
where $m(\cdot)$ is twice continuously differentiable, strictly increasing, strictly concave, and WLOG $m(0) = 0$. I will not provide intuition alongside the results because while the maths becomes more complex with this additional non-linear function, there is no new economic insight.

One disadvantage of introducing this non-linear sub-utility is that a closed form characterisation of equilibrium behaviour is no longer possible. This makes the network effects far harder to see. Nevertheless, existence and uniqueness are unaffected.

\begin{rem}\label{rem:nonlinear equilibrium}
Suppose $\alpha_i < 1$ for all $i$. There is a unique Nash Equilibrium.
\end{rem}

\Cref{prop:ref strength} then goes through unaffected. 

\begin{prop}\label{prop:nonlinear ref strength}
If $\alpha_i < 1$ for all $i$: (i) $x_i^*$ is weakly increasing, and (ii) $u_i^*$ is weakly decreasing, in $\alpha_j$ for all $i,j$, and strictly so if $i=j$.
\end{prop}

\Cref{prop:cost} is subject to some minor modification because the threshold in part (ii) is difficult to characterise with non-linear $m(\cdot)$

\begin{prop}\label{prop:nonlinear cost}
If $\alpha_i < 1$ for all $i$: (i) $x_i^*$ is strictly decreasing in $c$, (ii) $u_i^*$ is strictly increasing in $c$ if $x_i^*$ is sufficiently sensitive to $c$ (specifically, if $\frac{d x_i^*}{d c} < \min \{ \frac{-x_i^*}{c} , \frac{m'(x_i^*)}{c \cdot m''(x_i^*)} \}$).
\end{prop}

Without a closed form solution for $x_i^*$, it is not possible to show the tight link between changes to the network structure and changes in equilibrium consumption (i.e. an analogue to \Cref{prop:perturbation}). Nevertheless, simply looking at the First Order Condition is suggestive.

\begin{align}\label{eq:nonlinear FOC}
m(x_i) - \alpha_i \sum_j g_{ij} m(x_j) = F \left( \frac{c}{m'(x_i)} \right).
\end{align}

We can see that agents' actions are complements. If one of $i$'s neighbours consumes more, this pushes up $i$'s reference point, which in turn pushes her to consume more in an effort to keep up. Intuitively, if $i$ shifts from comparing herself with low consumption friends to high consumption neighbours, then her reference point increases, and so she chooses to consume more. The more ``The Joneses'' consume, the more you need to do to keep up with them.

\subsection*{Proofs}

\begin{proof}[\emph{\textbf{Proof of Remark \ref{rem:nonlinear equilibrium}.}}]
Consider the First Order Condition [FOC] (with some minimal rearranging) -- \Cref{eq:nonlinear FOC}.
LHS is strictly increasing in $x_i$. $m'(x_i)$ is strictly decreasing in $x_i$, and $F(\cdot)$ is strictly decreasing in its argument. So RHS is strictly decreasing in $x_i$. Hence there is a unique value of $x_i$ that solves the First Order Condition. I use $x_i^*$ to denote \emph{equilibrium} consumption (i.e. when \emph{all} FOCs hold simultaneously). For clarity, let $\hat{x}_i$ denote the value that solves just agent $i$'s FOC, for given values of $x_j$, $j \neq i$. \\

\noindent \textbf{Existence.} Note that $\hat{x}_i$ is increasing in $x_j$. So if no agent wants to choose more than $A$ when all other agents are choosing $A$, then there cannot exist an equilibrium where any agent chooses $A$ or more.
% If everyone wants to earn less than the friends they compare themselves with, then such high actions are not sustainable -- everyone wants to move down a bit.
%Now notice that there exists some number $\hat{A}$ such that for any $A > \hat{A}$ and for all $i$: if $x_j = A$ for all $j \neq i$, then $\hat{x}_i < A$.
Suppose all other agents choose $x_j = A$. Then $i$ wants to choose less than $A$ if and only if: $m(A)[ 1 - \alpha_i] > F(c/m'(A))$.\footnote{Note that after we pulled $m(A)$ out of the summation, $\sum_j g_{ij} = 1$, and so disappears.} 
Since LHS is increasing in $A$, and RHS is decreasing in $A$, there must exist some finite value of $A$ such that this inequality is true for all $i$. Denote this $\hat{A}$. Given this, we only need to consider the \emph{compact} space $[0,\hat{A}]^n$.
The functions are continuous by assumption. Therefore Brouwer's Fixed Point Theorem (e.g. \cite{border1985fixed}) guarantees existence. \\

\noindent \textbf{Uniqueness.} Proof by contradiction. Denote the equilibrium with the smallest action $x^*$. Suppose there is another equilibrium, $x^{* \prime} \equiv x^* + D$ (it is convenient and WLOG to write the second equilibrium as the initial one, plus some change). Consider the agent $i$ whose action increases the most when moving from $x^*$ to $x^{* \prime}$. That is, the agent $i$ such that $D_i \geq D_j$ for all $j \neq i$. Solve the first order condition for this agent $i$, assuming that all other agents are playing the equilibrium $x^{* \prime}$. %Now consider $i = \argmax_{j \in N} \{ D_j \}$. By definition, all of $i$'s neighbours have a smaller increase in their action when moving from $x^*$ to $x^{*'}$.
This is, we solve:
\begin{align*}
    x_i - \alpha_i \sum_j g_{ij} [x_j^* + D_j]) &= F \left( \frac{c}{m'(x_i)} \right)
\end{align*}
WLOG let the solution takes the form $x_i^* + z_i$.
\begin{align*}
    x_i^* + z_i - \alpha_i \sum_j g_{ij} x_j^* - \alpha_i \sum_j g_{ij} D_j  &= F \left( \frac{c}{m'(x_i^* + z_i)} \right)
\end{align*}
Substituting in the solution to the FOC from the initial equilibrium $x^*$ (and rearranging):
\begin{align*}
   z_i - \alpha_i \sum_j g_{ij} D_j &= F \left( \frac{c}{m'(x_i^* + z_i)} \right) - F \left( \frac{c}{m'(x_i^*)} \right)
\end{align*}
Since $F(\cdot)$ is decreasing in its argument. 
\begin{align*}
    z_i - \alpha_i \sum_j g_{ij} D_j < 0
\end{align*}
This actually assumes that $z_i > 0$. But if $z_i \leq 0$ then we have already reached a contradiction, since we require the increase in $i$'s action to be larger than the increase in all other agents' actions (when moving from $x^*$ to $x^{* \prime}$).
Then applying the requirement that $D_i \geq D_j$ for all $j \neq i$:
\begin{align*}
    z_i < \alpha_i \sum_j g_{ij} D_j \leq \alpha_i \sum_j g_{ij} D_i %= \alpha_i D_i \sum_j g_{ij} 
    = \alpha_i D_i < D_i
\end{align*}
Contradiction.
\end{proof}

\begin{proof}[\emph{\textbf{Proof of Proposition \ref{prop:nonlinear ref strength}}}]
For some fixed parameters $g$, $\alpha$, $c$, $b$ and functions $f(\cdot)$, $m(\cdot)$ there is a unique equilibrium $x^*$. 
\textbf{(i)} Suppose that $\alpha_j \uparrow$ for some $j$. Then $\hat{x}_j \uparrow$ (i.e. $j$'s optimal action rises \emph{conditional on} all other agents' actions). In turn $\hat{x}_j \uparrow \implies \hat{x}_k \uparrow$ for all $k$ s.t. $g_{kj} > 0$. In turn this increases $x_i$ for all $i$ (but only weakly so, since there is no guarantee that there exists a directly path from $j$ to $i$).

\Cref{rem:nonlinear equilibrium} guarantees a unique equilibrium. So this process must eventually converge to an equilibrium. But clearly all equilibrium actions weakly rises, and strictly so for the ``initial'' agent (who experienced the increase in $\alpha_j$).
\textbf{(ii)} In equilibrium we have:
\begin{align*}
    m(x_i^*) - \alpha_i \sum_j g_{ij} m(x_j^*) = F \left( \frac{c}{m'(x_i^*)} \right)
\end{align*}
We can substitute this back into the utility function to find equilibrium utility.
\begin{align*}
    u_i^* = f \left( F \left[ \frac{c}{m'(x_i^*)} \right]  \right) - c x_i^*
\end{align*}
We know from (i) that $\alpha_j \uparrow$ weakly increases $x_i^*$. Finally, notice that $u_i^*$ is strictly decreasing in $x_i^*$. To see this clearly, note that $x_i^* \uparrow \implies m'(x_i^*) \downarrow \implies F(c/m'(x_i^*)) \downarrow \implies f(F(c/m'(x_i^*))) \downarrow$. 
\end{proof}

\begin{proof}[\emph{\textbf{Proof of Proposition \ref{prop:nonlinear cost}.}}]
For some fixed parameters $g$, $\alpha$, $c$, $b$ and functions $f(\cdot)$, $m(\cdot)$ there is a unique equilibrium $x^*$. 
\textbf{(i)} Suppose that $c \uparrow$. This decreases RHS of the FOC (\Cref{eq:nonlinear FOC}). To restore equality, it must be that $\hat{x}_i \downarrow$ (as this decreases LHS and increases RHS). In turn $\hat{x}_i \downarrow \implies \hat{x}_j \downarrow$. \Cref{rem:nonlinear equilibrium} guarantees a unique equilibrium. So this process must eventually converge to an equilibrium. But all equilibrium actions strictly decrease.
\textbf{(ii)} Equilibrium utility is:
\begin{align*}
    u_i^* = f \left( F \left[ \frac{c}{m'(x_i^*)} \right]  \right) - c x_i^*.
\end{align*}
This was derived in the proof to \Cref{prop:nonlinear ref strength}. It is clear that a sufficient condition for $u_i^*$ to be strictly increasing in $c$ is that: (a) $c x_i^*$ and (b) $c/m'(x_i^*)$ are strictly decreasing in $c$ (recall that $f(F(\cdot))$ is decreasing in its argument)\footnote{So a decrease in $c/m'(x_i^*)$ increases $f(F(c/m'(x_i^*)))$.} 
Consider each condition separately. First,
\begin{align*}
    \frac{d (c x_i^*)}{d c} = c \cdot \frac{d x_i^*}{d c} + x_i^*,
\end{align*}
which is less that zero if and only if $\frac{d x_i^*}{d c} < - x_i^*/c$. Second, 
\begin{align*}
    \frac{d ( c/m'(x_i^*) )}{d c} = [m'(x_i^*)]^{-1} - c m''(x_i^*) [m'(x_i^*)]^{-2} \frac{d x_i^*}{d c}
\end{align*}
which is less than zero if and only if $\frac{d x_i^*}{d c} < - m'(x_i^*)/[c \cdot m''(x_i^*)]$.
\end{proof}
An equivalent sufficient condition for $u_i^*$ to be strictly decreasing in $c$ would be easy. The method for proving it would be similar to above. The condition would simply be  $\frac{d x_i^*}{d c} > \max \{ \frac{-x_i^*}{c} , \frac{m'(x_i^*)}{c \cdot m''(x_i^*)} \}$. 
Due to the inability to find closed form solutions for $x_i^*$ when $m(\cdot)$ is nonlinear, I am not able to provide a tight link between network centrality and how welfare is affected by cost changes. Nevertheless, the flavour of \Cref{prop:cost} goes through -- agent's whose consumption is highly sensitive to costs end up experiencing welfare gains when the marginal cost of consumption rises. This is again because the benefits (in terms of a lower reference point) of lower consumption by neighbours more than offsets the directly higher cost of consumption.

%% file: D.loss_aversion.tex
In \cref{sec:model} we assume that $u_i = f(x_i - \alpha_i \sum_j g_{ij} x_j)$, where $f(\cdot)$ is strictly increasing and concave. This captures references dependence and diminishing sensitivity, but \emph{not} loss aversion. To capture loss aversion, we now assume that $f(\cdot)$ is strictly concave [resp. convex] in the positive [resp. negative], and kinked at zero. This is in line with the Kahneman and Tversky's canonical setup \citep{tversky1979prospect, tversky1991loss}. Formally, we can write these properties as the following restrictions on the function $f(\cdot)$: (i) $f:a \to \mathbb{R}$ for $a \in \mathbb{R}$, (ii) $f'(a) > 0 \ \forall a$, (iii) $f''(a) < 0 \ \forall a > 0$, $f''(a) > 0 \ \forall a < 0$, (iv) $\lim_{a \to 0^{-} } | f'(a) | > \lim_{a \to 0^{+} } | f'(a) |$, (v) $f''(a) \in \mathbb{R} \ \forall a \neq 0$, $f(0) = 0$.

Adding the functional form requirements for loss aversion does not affect \Cref{prop:equilibrium}. This is because agents are always able to, and always choose to, consume above their reference point in this model. Therefore the convexity in the negative domain that is the core addition of loss averse preferences (over and above reference dependent preferences) has no bite.

%% file: E.many_goods.tex
In \cref{sec:model} we assume that there is only 1 good. Here, consider the model with $K$ goods (i.e. a K-dimensional consumption bundle), as in \cite{koszegi2006model}. This generalisation has no effect.

There are $K$ goods, $x_1,...,x_K$. Each agent $i$ simultaneously chooses a consumption bundle $(x_{i1}, x_{i2}, ... , x_{iK}) \in \mathbb{R}^{n}_{+}$. Following \cite{koszegi2006model}, I assume that preferences are additively separable over goods. All other elements of the model are the same as in \cref{sec:model}. Therefore $i$'s utility function is:
\begin{align}
    u_{i} = \sum_{k=1}^{K} \left[ f \left( x_{ik} - \sum_{j} \alpha_{i} g_{ij} x_{jk} \right) - c_k x_{ik} \right] + b_i \sum_j \alpha_i g_{ij}
\end{align}
where $f(\cdot)$ has the same properties as in \cref{sec:model}. Now consider agents' First Order Conditions.
\begin{align}
    \frac{d u_i}{d x_{ik}} = f'\left( x_{ik} - \sum_{j} \alpha_{i} g_{ij} x_{jk} \right) - c_k = 0 \ \text{ for all } i \ \text{ and for all } k
\end{align}
These First Order Conditions are clearly identical to those in the 1-good case. Obviously there are now $K$ First Order Conditions for each agent, but $x_{ik}$ only appears in the First Order Conditions relating to good $k$, and never in any relating to $k' \neq k$. Therefore the solution for each good $k$ is the same as it would be if $k$ were the only good.

%% file: F.network_structure.tex
This section presents technical results regarding the effects of comparison shifts, and formalises claims made in section \ref{sec:results}. The natural starting point is to characterise the exact impact of a single comparison shift.

\begin{prop}\label{prop:perturb effect restated}
	The change in agent $i$'s consumption due a comparison shift $D$, $\Delta x_i^*$, is equal to:
	$$\frac{ \phi B_{ir} (C^{b}_{u} - C^{b}_{d}) }{1 - \phi (B_{ur} - B_{dr}) } \cdot f'^{-1}(c)$$
\end{prop}
\begin{quote}
\begin{proof}[\emph{\textbf{Proof of Proposition \ref{prop:perturb effect restated}}}]
This is a restatement of \Cref{lem:perturb effect} (which was needed to prove other results from the main text).
\end{proof}
\end{quote}

This exact characterisation is not very user-friendly and is largely technical. However, it is useful because it forms the basis for a number of further results, most importantly \Cref{prop:perturbation}. Its implications for the effects of a comparison shift on different agents are immediate.

% RESULT STATES TWO IMPLICATIONS OF AN ELEMENTARY PERTURBATION
\begin{cor}\label{cor:perturb cor1}
	Given a comparison shift $D$, the change in optimal actions is: \\
	(i) in the same direction for all agents, \\
	(ii) proportional to the amount that agent $i$ compares herself to $r$ (the subject of the comparison shift) prior to the shift.
\end{cor}
% PROOF OF REMARK IMMEDIATELY ABOVE
\begin{quote}
\begin{proof}[\emph{\textbf{Proof of Corollary \ref{cor:perturb cor1}.}}]
Proposition \ref{prop:perturb effect restated} characterises the change in agent $i$'s optimal action, $\Delta x_{i}^{*}$. The only term in the expression for $\Delta x_i^*$ that depends on $i$ is $B_{ir}$. \textbf{(i)} Since $B_{ij} \geq 0$ for all $i,j$, $\Delta x_{i}^{*}$ has the same sign for all $i$. \textbf{(ii)} $\Delta x_{i}^{*}$ is equal to $B_{ir}$ multiplied by some terms that do not depend on $i$. Therefore $\Delta x_{i}^{*}$ is proportional to $B_{ir}$ (i.e. the amount that $i$ compares herself to $r$).
%Differentiating $\Delta x_{i}^{*}$ with respect to $B_{ir}$ yields:
%\begin{align*}
%\frac{d(\Delta x_{i}^{*})}{d B_{ir}} = \frac{\phi(C^{b}_{u} - C^{b}_{d}) f'^{-1}(c) }{1 - \phi (B_{ur} - B_{dr}) }
%\end{align*}
%Which does not contain any terms that depend on $i$. Therefore, $\Delta x_{i}^{*}$ is linear in (i.e. proportional to) $B_{ir}$.
\end{proof}
\end{quote}

It also allows us to examine the nature of the returns to the magnitude of a comparison shift. When an agent $r$ moves from direct comparison with a lower centrality agent to direct comparison with a higher centrality agent (i.e. $C^{b}_{u} > C^{b}_{d}$), then there are increasing returns to the magnitude of a comparison shift if and only if agent $r$ influences agent $u$ more than she influences agent $d$, but the gap is not too large.

\begin{cor}\label{cor:perturb cor2}
	Given a comparison shift $D$, if $C^{b}_{u} > C^{b}_{d}$, then for all $i$, $x_i^*$ is: convex in $\phi$ if and only if $B_{ur} - B_{dr} \in (0, \frac{1}{\phi})$ , and concave in $\phi$ otherwise.
\end{cor}
\begin{quote}
\begin{proof}[\emph{\textbf{Proof of Corollary \ref{cor:perturb cor2}.}}]
For ease of exposition, let: $ B_{ir} (C^{b}_{u} - C^{b}_{d}) f'^{-1}(c) \equiv y$ and $(B_{ur} - B_{dr}) \equiv z$. Let $\hat{x}_i^* \equiv x_i^* + \Delta x_i^*$ be the new equilibrium value of consumption. So $\hat{x}_i^* = x_i^* + y \phi (1 - z \phi)^{-1}$.
Now consider the second derivative. $\frac{d^2 \hat{x}_i^* }{d \phi^2} = 2 yz (1 - z \phi)^{-3}$. We can partition values of $z$ into three cases: (i) $z < 0$, (ii) $0 \leq z \leq \frac{1}{\phi}$, (iii) $z > \frac{1}{\phi}$. In case (ii), $\frac{d^2 \hat{x}_i^*}{d \phi^2} \geq 0$. In cases (i) and (iii), $\frac{d^2 \hat{x}_i^*}{d \phi^2} < 0$. These observations follow straightforwardly from the fact that $y >0$ and $\phi > 0$. The second derivative determines convexity/concavity. Note that if $(C^{b}_{u} - C^{b}_{d}) < 0$ then $y < 0$ and so all results flip.
\end{proof}
\end{quote}

Now consider a \emph{composite comparison shift}, $\widehat{D} = D_1 + ... + D_Z$, which is constructed by summing up $n \geq 2$ comparison shifts. The impact of a composite comparison shift is not equal to the sum of the effects of each comparison shift that form part of it. This is because the impact of any one comparison shift depends on the \emph{whole} network immediately before it occurs.

% RESULT STATES THAT YOU CAN'T SUM UP EFFECTS OF ELEMENTARY PERTURBATIONS TO GET EFFECT OF A COMPOSITE PERTURBATION
\begin{cor}\label{cor:perturb cor3}
	Consider a composite comparison shift $\widehat{D}$. The impact of $\widehat{D} = D_1 + ... + D_Z$ is not the same as the sum of the impacts of $D_1 , ... , D_Z$.
\end{cor}
% PROOF OF REMARK IMMEDIATELY ABOVE
\begin{quote}
\begin{proof}[\emph{\textbf{Proof of Corollary \ref{cor:perturb cor3}.}}]
W.L.O.G any composite comparison shift $\hat{D}$ can be expressed as the sum of comparison shifts $D_1,..., D_Z$: $\hat{D} \equiv \sum_{i=1}^{Z} D_i$. It follows from Lemma \ref{perturbation theorem} that $H$ is a function of both the comparison shifts $D_i$ and $(I-G)^{-1}$.

Trivially $(I-G-X)^{-1} \neq (I-G)^{-1}$ for any $X \neq 0$. Therefore the effect of a comparison shift $D_i$ depends on the network immediately prior to the comparison shift: $H(D_i, (I - G - \sum_{j=1}^{i-1} D_i)^{-1}) \neq H(D_i , (I-G)^{-1})$. The effect of a composite comparison shift $\hat{D}$ is therefore not equal to the sum of the effects of individual comparison shifts: $\sum_{i=1}^{Z} H(D_i, (I - G - \sum_{j=1}^{i-1} D_i)^{-1}) \neq \sum_{i=1}^{Z} H(D_i , (I-G)^{-1})$.
\end{proof}
\end{quote}

For a given network, it would be straightforward to calculate the change in actions due to a more complex change in the reference structure using an application of the formula provided by \cite{chang2006inversion}. However, we cannot obtain analytic results due to the interactions between the effects of each comparison shift. Even with a few comparison shifts, the outcome would be too complicated to yield any insight. However, if the comparison shifts are all of an equal size, for example $\Delta$, then interactions between the are of the order $\Delta^{2}$. Therefore, when considering small changes to the network (i.e. small $ Z \cdot \Delta$) we can reasonably disregard the interactions -- a naive summation is a close approximation for the actual aggregate effect.

% RESULT STATES THAT IF TOTAL CHANGE TO NETWORK IS SMALL, THEN THE NAIVE APPROXIMATION WORKS
\begin{cor}\label{cor:perturb cor4}
	If the total change to the network, $Z \cdot \Delta$, is small,then the impact of a composite comparison shift $\widehat{D}$ approximately equal to the sum of the impacts of each comparison shift that makes up $\widehat{D}$.
\end{cor}
% PROOF OF REMARK IMMEDIATELY ABOVE
\begin{quote}
\begin{proof}[\emph{\textbf{Proof of Remark \ref{cor:perturb cor4}.}}]
Let the composite comparison shift $\hat{D}$ be constructed from a series of comparison shifts $(D_1 , ... , D_Z)$, each of a size $\Delta$. $H(D_i , (I-G-X)^{-1})$ is the change in the optimal actions following a comparison shift $D_i$, when the starting network is $(G+X)$. From Lemma \ref{perturbation theorem}: $H(D_i, (I-G)^{-1}) = \Delta \cdot fn((I-G)^{-1}) = \mathcal{O}(\Delta)$. Then $H(D_i, (I-G-D_j)^{-1}) = \Delta \cdot fn((I-G-D_j)^{-1}) = \Delta \cdot fn((I-G)^{-1} + \mathcal{O} (\Delta) ) = H(D_i, (I-G)^{-1}) + \mathcal{O}(\Delta^2)$. By a simple induction argument we can see that $H(D_i, (I-G- \sum_{j=1}^{J} D_j)^{-1}) = H(D_i, (I-G)^{-1}) + \mathcal{O}(\Delta^2)$. Therefore the effect of the earlier comparison shifts $D_1,...,D_{i-1}$ on the change in optimal actions induced by $D_i$ is on the order $\Delta^2$. If $Z \cdot \Delta$ is sufficiently small, then we can ignore these interaction effects, which are collectively of the order $Z \cdot \Delta^2$.
\end{proof}
\end{quote}

This approach works because we are able to take a linear approximation (i.e. ignore any terms that are $\mathcal{O}(\delta^{2})$). The total change to agents' Bonacich centrality will be well approximated by simply adding up the effects of each comparison shift. However, we should be very wary of extrapolating from small changes to large ones. It is difficult to characterise the interactions and small changes may not be indicative of large ones. Past observations from small or localised changes may cease to be a useful guide in the face of large-scale social change happens.

Finally, we prove that assuming $\alpha_i < 1$ for all $i$ is sufficient to guarantee existence, even after a comparison shift (something we claimed, but did not prove, in section \ref{sec:results}).

% RESULT STATES THAT ALPHA < 1 FOR ALL AGENTS ENSURES EXISTENCE EVEN AFTER A PERTURBATION
\begin{rem}\label{perturb corr 5}
If $\alpha_{i} < 1 \ \forall i \in N$ then $(I - (ag+D))^{-1}$ exists for any network $ag$, and any comparison shift $D$.
\end{rem}
\begin{quote}
% PROOF OF REMARK IMMEDIATELY ABOVE
\begin{proof}[\emph{\textbf{Proof of Remark \ref{perturb corr 5}.}}]
This follows trivially from Proposition \ref{prop:equilibrium}, which proves that $\alpha_{i} < 1$ for all $i$ is a sufficient condition to ensure equilibrium existence for any network $G$. Since a comparison shift leaves $\alpha_{i}$ unchanged for all $i$ (by definition), then Proposition \ref{prop:equilibrium} continues to apply.
\end{proof}
\end{quote}

However, if $\alpha_{i} > 1$ for some $i$, then it is possible that a solution exists for some, \emph{but not all}, reference structures $g$. In this instance it is necessary to check that a solution exists both before and after the shift. That is: $\lambda_1(G) < 1$ and $\lambda_1(G + D) < 1$. While it is possible for only one of these conditions to be met, the results cannot apply unless both hold.

%% file: K.Chang_2006.tex
I restate \Cref{perturbation theorem} for clarity and then present a proof. This is a special case of \cite{chang2006inversion}, which simplifies the proof. I have also aligned the notation to match my paper.

% ------------------------------------------------------------- %
% LEMMA: CHARACTERISES THE EXACT IMPACT OF AN ELEMENTARY PERTURBATION
% ------------------------------------------------------------- %
\begin{lem*}[A perturbation theorem. \citet{chang2006inversion}] %\label{perturbation theorem restate}
If $D$ is a comparison shift, then $(I - [G + D])^{-1} - (I - G)^{-1} = H$. Where:
%\begin{align*}
%H = \frac{\phi}{1 - \phi (B_{ur} - B_{dr}) }
%	\begin{pmatrix}
%	B_{1r} (B_{u1} - B_{d1}) & B_{1r} (B_{u2} - B_{d2}) & \cdots & B_{1r} (B_{un} - B_{dn}) \\
%	B_{2r} (B_{u1} - B_{d1}) & B_{2r} (B_{u2} - B_{d2}) & \cdots & B_{2r} (B_{un} - B_{dn}) \\
%	\vdots & \vdots & \ddots & \vdots \\
%	B_{nr} (B_{u1} - B_{d1}) & B_{nr} (B_{u2} - B_{d2}) & \cdots & B_{nr} (B_{un} - B_{dn})
%	\end{pmatrix}
%\end{align*}
% SHOW A SMALLER VERSION
\begin{align*}
H = \frac{\phi}{1 - \phi (B_{ur} - B_{dr}) }
	\begin{pmatrix}
	B_{1r} (B_{u1} - B_{d1}) & \cdots & B_{1r} (B_{un} - B_{dn}) \\
	\vdots &  \ddots & \vdots \\
	B_{nr} (B_{u1} - B_{d1}) & \cdots & B_{nr} (B_{un} - B_{dn})
	\end{pmatrix}
\end{align*}
\end{lem*}
% ------------------------------------------------------------- %
% PROVES THE LEMMA IMMEDIATELY ABOVE
% ------------------------------------------------------------- %
\begin{proof}[\emph{\textbf{Proof of Lemma \ref{perturbation theorem}.}}]
This proof is a simplified version of \citet{chang2006inversion}. By the Woodbury Identity Matrix: if $A$ and $D$ are $n \times n$ matrices, and $A$ is non-singular, then $(A - D)^{-1} = E + E (I - DE)^{-1} DE$, where $E \equiv A^{-1}$ [Woodbury (1950) and Sherman and Morrison (1949)]. Now partition the matrices $D$ and $E$:
\begin{align*}
D =
	\begin{bmatrix}
	\overline{\underline{D}} & 0 \\
	0 & 0 \\
	\end{bmatrix}
\qquad
E =
	\begin{bmatrix}
	\overline{\underline{E}} & E_{2} \\
	E_{1} & E_{3} \\
	\end{bmatrix}
\qquad
\overline{E} =
	\begin{bmatrix}
	\overline{\underline{E}} \\
	E_{1}
	\end{bmatrix}
\qquad
\underline{E} =
	\begin{bmatrix}
	\overline{\underline{E}} & B_{E} \\
	\end{bmatrix}
\end{align*}
where $\overline{\underline{D}}$ is the smallest matrix that contains non-zero elements of $D$, and $\overline{\underline{E}}$ contains the transpose of the elements in $\overline{\underline{E}}$.\footnote{So if $\overline{\underline{D}}$ consists of elements $\overline{\underline{D}}_{ij}$ for $i \in I \ , \ j \in J$, then $\overline{\underline{D}}$ consists of elements $\overline{\underline{D}}_{ji}$ for $i \in I \ , \ j \in J$.}
Simple matrix algebra yields: $(A - D)^{-1} = E + \overline{E} (I - \overline{\underline{D}} \ \overline{\underline{E}})^{-1} \overline{\underline{D}} \ \underline{E}$.

\noindent Now recall that $D$ is a `comparison shift' (as per Definition \ref{defn: elementary perturbation}). Therefore: $\overline{\underline{D}} = \begin{bmatrix} D_{ru} & D_{rd} \end{bmatrix} = \begin{bmatrix} \phi & - \phi \end{bmatrix}$, and so:
\vspace{-8mm}
\begin{align*}
\overline{\underline{E}} &=
	\begin{bmatrix}
	E_{ur} \\
	E_{dr}
	\end{bmatrix}
\ , \  \overline{E} =
	\begin{bmatrix}
	E_{1r} \\
	\vdots \\
	E_{nr}
	\end{bmatrix}
\text{ and } \underline{E} =
	\begin{bmatrix}
	E_{u1} & \cdots & E_{un} \\
	E_{d1} & \cdots & E_{dn} \\
	\end{bmatrix}
	\end{align*}
\vspace{-8mm}
\noindent This yields;
\begin{align*}
(A - D)^{-1} &= E +
	\begin{bmatrix}
	E_{1r} \\
	\vdots \\
	E_{nr}
	\end{bmatrix}
\cdot
	\left( I -
	\begin{bmatrix} D_{ru} & D_{rd} \end{bmatrix}
	\begin{bmatrix}
	E_{ur} \\
	E_{dr}
	\end{bmatrix}  \right)^{-1}
\cdot
	\begin{bmatrix} D_{ru} & D_{rd} \end{bmatrix}
\cdot
	\begin{bmatrix}
	E_{u1} & \cdots & E_{un} \\
	E_{d1} & \cdots & E_{dn} \\
	\end{bmatrix}
\end{align*}
Now multiply out, substitute in $D_{ru} = \phi$ and $D_{rd} = - \phi$, noticing that the matrix inverse $(I - \overline{\underline{D}} \ \overline{\underline{E}})^{-1}$ is a scalar, and rearrange.
\vspace{-6mm}
%\begin{align*}
%(A-D)^{-1} - A^{-1} = \frac{\phi}{1 - \phi (B_{ur} - B_{dr}) }
%	\begin{bmatrix}
%	B_{1r} (B_{u1} - B_{d1}) & B_{1r} (B_{u2} - B_{d2}) & \cdots & B_{1r} (B_{un} - B_{dn}) \\
%	B_{2r} (B_{u1} - B_{d1}) & B_{2r} (B_{u2} - B_{d2}) & \cdots & B_{2r} (B_{un} - B_{dn}) \\
%	\vdots & \vdots & \ddots & \vdots \\
%	B_{nr} (B_{u1} - B_{d1}) & B_{nr} (B_{u2} - B_{d2}) & \cdots & B_{nr} (B_{un} - B_{dn})
%	\end{bmatrix}
%\end{align*}
% SHOW A SMALLER VERSION
\begin{align*}
(A-D)^{-1} - A^{-1} = \frac{\phi}{1 - \phi (B_{ur} - B_{dr}) }
	\begin{pmatrix}
	B_{1r} (B_{u1} - B_{d1}) & \cdots & B_{1r} (B_{un} - B_{dn}) \\
	\vdots & \ddots & \vdots \\
	B_{nr} (B_{u1} - B_{d1}) & \cdots & B_{nr} (B_{un} - B_{dn})
	\end{pmatrix}
\end{align*}
\noindent This result holds for any non-singular $n \times n$ matrix $A$. Letting $A = (I - G)$ yields the result.
\end{proof}
% ------------------------------------------------------------- %

%% file: G.heterogeneous_costs.tex
The cost parameter $c_{i}$ reflects agents' underlying propensity/ability to take the action, $x_{i}$. In this section, we examine the implications of introducing cost heterogeneity. The key takeaway is that allowing for heterogeneous costs has relatively little impact on the insights of the main model.

With this generalisation, equilibrium play depends on the individual entries in the full matrix $B$, rather than only on Bonacich centralities (the row sums of $B$). Nevertheless, the existence condition is unaffected, and the solution takes a similar form. 
At this point, it is helpful to interpret individual elements of the matrix $B$ and to introduce a notion of ``generalised Bonacich centrality''. 

\begin{defn}[Comparisons]
The comparison matrix is $B \equiv (I - G)^{-1}$, where $B_{ij}$ captures how much $i$ compares herself to $j$.
\end{defn}

Bonacich centrality captures an agent's connectedness to the network as a whole. The comparison matrix breaks this down to the individual agent level. An element $B_{ij}$ measures the total weight of walks from $i$ to $j$, and captures the extent to which $i$ compares herself to $j$. This is a dis-aggregation of Bonacich centrality.
We can then weight these individual level comparisons with a function of the cost parameters to obtain a generalised notion Bonacich centrality.\footnote{It is clear that when $c_j = c$ for all $j$ this collapses back to the original Bonacich centrality.}

\begin{defn}[Generalised Bonacich Centrality]\label{general bonacich centrality}
The vector of Bonacich centralities for a network $G \equiv a g$ is $C^{gen} = B \cdot f^{'-1}(c)$ The Bonacich centrality of agent $i$ is $C^{gen}_{i} = \sum_{j} B_{ij} f^{'-1}(c_{j})$
\end{defn}
With this definition we can restate \Cref{prop:equilibrium}, accounting for heterogeneous costs. The condition for existence depends only on the network and so is unaffected by cost heterogeneity. However, the optimal actions are now proportional to our new notion of generalised centrality, rather than the usual Bonacich centralities.

% RESULT EXTENDS THE EXISTENCE/UNIQUENESS/SOLUTION TO THE HETEROGENEOUS COSTS CASE
\begin{prop}[Existence and Solution]\label{rem:solution HC}
With heterogeneous costs, an equilibrium exists if and only if $\lambda_{1} < 1$. If this condition is met, then there is a unique Nash Equilibrium: $x_{i}^{*} = C^{gen}_{i}$
	%\begin{align*}
	%x_{i}^{*} = C^{gen}_{i}
	%\end{align*}
\end{prop}
% PROOF OF REMARK IMMEDIATELY ABOVE
\begin{quote}
\begin{proof}
This follows from the proofs to Proposition \ref{prop:equilibrium}, with the only change that $f^{'-1}(c_{j})$ now depends on $j$, and so cannot be pulled out through the summation sign. %Noticing that $C^{gen}_{i} = \sum_{j} B_{ij} f^{'-1}(c_{j})$ by definition completes the argument.
\end{proof}
\end{quote}

The other results from \Cref{sec:results} also extend to this heterogeneous cost setting. The proofs here only provide the required extension from their homogeneous cost analogues.

\begin{prop}[Reference strength]
If $\lambda_1 < 1$: (i) $x_{i}^{*}$ is weakly increasing, and (ii) $u_i^*$ is weakly decreasing, in $\alpha_{j}$ for all $i,j$, and strictly so if $i=j$.
\end{prop}
\begin{quote}
\begin{proof}
\textbf{(i)} \Cref{prop:ref strength} shows that if a given $\alpha_{j}$ increases, then all elements of $(ag)^{k}$ weakly increase for any $k$. Consequently, all elements $(I - ag)^{-1}_{ij}$ increase. %Therefore $x_i^* = \sum_j (I - ag)^{-1}_{ij} f^{'-1}(c_{j})$ increases. 
\textbf{(ii)} having found that $x_i^*$ is increasing in $\alpha_k$, the second part of the proof to \Cref{prop:ref strength} goes through unchanged.
\end{proof}
\end{quote}

\begin{prop}[Cost]
If $\lambda_1 < 1$: (i) $x_i^*$ is strictly decreasing and convex in $c_j$ for all $j$, (ii) $u_i^*$ is strictly increasing in $c_j$ for all $j \neq i$.
\end{prop}
\begin{quote}
\begin{proof}
(i) follows straightforwardly from \Cref{rem:solution HC} and the fact that $F(\cdot)$ is strictly decreasing and convex. (ii) equilibrium utility is $u_i^* = f(x_i^* - \sum_{k \neq i} G_{ik} x_k^*) - c_i x_i^*$. Differentiate with respect to $c_j$:
\begin{align*}
    \frac{d u_i^*}{d c_j} = \left( \frac{d x_i^*}{d c_j} - \sum_{k \neq i} G_{ik} \frac{d x_k^*}{d c_j} \right) f'(\cdot) - c_i \frac{d x_i^*}{d c_j}
\end{align*}
Now recall that $f'(x_i^* - \sum_j x_j^*) - c_i = 0$ in equilibrium, and that $x_i^* = C_i^{gen}$ for all $i$, $\frac{d x_k^*}{d c_j} = \frac{d C_k^{gen}}{d c_j} = B_{kj} F'(c_j) < 0$. Substituting these in yields:
\begin{align*}
\frac{d u_i^*}{d c_j} = - f'(\cdot) \sum_{k \neq i} G_{ik} B_{kj} F'(c_j) < 0
\end{align*}
\end{proof}
\end{quote}
This result is somewhat different to the homogeneous cost version. Because the cost parameter is now agent-specific, the outcome is much simpler. An increase in an agent $j$'s cost pushes down her consumption, relaxing the need for others to keep up with ``The Joneses'' (in this case, agent $j$). This increases welfare for all $i \neq j$. Since $i$ has not experienced an increase in her own costs, there is no off-setting effect. The impact of someone else's cost parameter on your welfare is unambiguous.

\begin{prop}\label{rem:perturb effect HC}
The change in agent $i$'s action due a comparison shift $D$ is equal to;
\begin{align*}
\frac{ \phi B_{ir} (C^{gen}_{u} - C^{gen}_{d}) }{1 - \phi (B_{ur} - B_{dr}) }
\end{align*}
\end{prop}
%\begin{quote}
\begin{proof}
This follows from Proposition \ref{prop:perturb effect restated} but replacing $C^{b}_{i} \cdot f^{'-1}(c)$ with $C^{gen}_{i}$.
\end{proof}
%\end{quote}

All other results concerning comparison shifts also follow as a result of this, including an analogue to \Cref{prop:perturbation}. This is because they are also based on \Cref{prop:perturb effect restated} in the homogeneous cost case. 

\begin{prop}[Endogenous network]
In all pairwise stable networks, if $b_i \geq c_i f^{'-1}(c_i)$, then $G_{ij} > 0$ only if $\frac{b_i}{c_i} = \frac{b_j}{c_j}$.
\end{prop}
\begin{quote}
\begin{proof}
This follows straightforwardly from the proof to \Cref{prop: endogenous network}, replacing $c$ with the agent-specific version as appropriate.
\end{proof}
\end{quote}